\documentclass[10pt]{article}
\usepackage[letterpaper, margin=1in]{geometry}  
\usepackage{indentfirst}
\usepackage{amsmath}
\usepackage{amssymb}
\usepackage{amsfonts}
\usepackage{amsthm}
\usepackage{url}
\usepackage{hyperref}
\usepackage[all]{xy}
\usepackage{tikz}
\usepackage{stmaryrd}
\usepackage{mathrsfs}
\usepackage[mathcal]{eucal}
\usepackage{multirow}
\usepackage{slashed}
\usepackage{graphicx}
\usepackage{placeins}

\usepackage{mathtools}
\usepackage{tikz-cd}
\usepackage{tqft}
\usepackage{empheq}
\usepackage{enumitem}
\usepackage{footnote}
\makesavenoteenv{tabular}
\makesavenoteenv{table}
\usepackage{multirow}
\makeatletter
\renewcommand\subsubsection{\@startsection{subsubsection}{3}{\z@}%
  {-3.25ex\@plus -1ex \@minus -.2ex}%
  {0pt}%
  {\normalfont\normalsize\bfseries}}
\makeatother

\newtheoremstyle{spaced}
  {1ex}   
  {1ex}   
  {\normalfont} 
  {}         
  {\bfseries} 
  {.}        
  { }       
  {}         

\theoremstyle{spaced}

\newtheorem{thm}{Theorem}[section]
\newtheorem*{thm*}{Theorem}
\newtheorem*{rmk*}{Remark}
\newtheorem{lemma}[thm]{Lemma}
\newtheorem{prop}[thm]{Proposition}

\newtheorem{defn}[thm]{Definition}

\newtheorem{ex}[thm]{Example}

\newtheorem{rmk}[thm]{Remark}

\newcommand{\del}{\partial}

\newcommand*\circled[1]{\tikz[baseline=(char.base)]{
            \node[shape=circle,draw,inner sep=2pt] (char) {#1};}}

\newcommand{\abracket}[1]{\left\langle#1\right\rangle}

\newcommand{\bracket}[1]{\left(#1\right)}

\newcommand{\bb}[1]{\mathbb{#1}}
\newcommand{\mr}[1]{\mathrm{#1}}
\newcommand{\mc}[1]{\mathcal{#1}}

\newcommand{\mf}[1]{\mathfrak{#1}}

\newcommand{\ol}{\overline}
\newcommand{\ul}{\underline}

\newcommand{\eps}{\varepsilon}

\newcommand{\R}{\bb{R}}
\newcommand{\C}{\bb{C}}
\newcommand{\Z}{\bb{Z}}

\DeclareMathOperator{\umap}{\underline{\mr{Map}}}

\DeclareMathOperator{\uhom}{\underline{\mr{Hom}}}

\DeclareMathOperator{\HH}{HH}

\DeclareMathOperator{\Loc}{Loc}

\DeclareMathOperator{\IIB}{IIB}
\DeclareMathOperator{\IIA}{IIA}

\DeclareMathOperator{\coh}{{Coh}}

\DeclareMathOperator{\Hom}{Hom}
\DeclareMathOperator{\Ext}{Ext}

\DeclareMathOperator{\ext}{Ext}

\DeclareMathOperator{\higgs}{Higgs}

\DeclareMathOperator{\tot}{tot}

\DeclareMathOperator{\id}{id}

\DeclareMathOperator{\tr}{Tr}
\DeclareMathOperator{\pt}{pt}

\DeclareMathOperator{\PV}{PV}
\DeclareMathOperator{\GL}{GL}

\DeclareMathOperator{\U}{U}

\begin{document}

\title{Twists of Supersymmetric Yang--Mills Theory\\ in Topological String Theory}

\author{Philsang Yoo}

\date{}

\maketitle

\abstract{We show that every twist of pure supersymmetric Yang--Mills theory with gauge group $\GL(N)$ can be realized as an open-string field theory in topological string theory. Our approach reinterprets twists of supersymmetric Yang--Mills theory as generalized Chern--Simons theories, and identifies topological string backgrounds that produce the corresponding Chern--Simons theories.}


\section{Introduction}

D-branes are central objects in string theory, commonly described as extended objects on which open strings can end. At low energies, the dynamics on their world-volume are effectively captured by gauge theories arising from the massless sector of open-string field theory, with features determined by both the brane configuration and the surrounding string background. In the simplest cases, a stack of $N$ D$k$-branes on a flat space in type II string theories yields $(k+1)$-dimensional maximally supersymmetric Yang--Mills theory with gauge group $\U(N)$; more intricate brane configurations and nontrivial backgrounds, however, lead to various supersymmetric Yang--Mills theories in different dimensions. These brane realizations have provided profound insights into supersymmetric gauge theory, including connections between the planar limits of gauge theories and tree-level string theories, as well as dualities between different gauge theories.

A complementary approach to understanding the structure of supersymmetric field theories, including but not limited to gauge theories, is twisting. Twisting involves selecting invariants under a chosen supercharge, resulting in simplified field theories that depend only on the topological or holomorphic structure of spacetime. This idea has been pivotal in physical mathematics, notably in mirror symmetry and the formulation of Seiberg--Witten invariants\footnote{In fact, twisting of 4-dimensional $\mathcal{N}=2$ theories provides a physical context for the Donaldson invariants, and a weakly coupled abelianized description of the theory in the IR yields the Seiberg--Witten invariants.} for 4-manifold theory. Recently, twists of supersymmetric Yang--Mills theories have been systematically studied and completely classified in \cite{ESW21, ESW}.

These two distinct perspectives converge on a natural question: can every twist of a supersymmetric Yang--Mills theory be realized within string theory? A hint comes from \cite{WittenCSstring}, which argues that (topological) Chern--Simons theory can be realized as an open-string field theory of topological strings. Indeed, in the context of topological strings, we answer the question as follows:

\begin{thm*}
Every twist of pure supersymmetric Yang--Mills theory with gauge group $\GL(N)$ can be realized as an open-string field theory of topological strings, as summarized in the following tables.
\end{thm*}

\begin{rmk*}
Costello and Li \cite{CL16} have extended the idea of twisting to string theory and supergravity, and conjectured that certain twists of type II string theories are, in fact, topological string theories (see Remark \ref{rmk:CLconjecture} and the discussion above). We view our main theorem as supporting evidence for this conjecture, since the twisted versions of supersymmetric Yang--Mills theories naturally emerge as world-volume gauge theories in these topological strings --- mirroring the way physical supersymmetric Yang--Mills theories arise in physical string theory.
\end{rmk*}

\begin{table}[htbp]
\centering
\renewcommand{\arraystretch}{0.8}
\begin{tabular}{c|c|c||c|c||c|c}
{\(d\)} & \( {\mathcal{N}}\)    &  {twist}&  {spacetime}& \( {A}\)&  {D-branes}&  {closed string} \\
\hline
\multirow{2}{*}{10} 
 & \multirow{2}{*}{$(1,0)$}       
 & \multirow{2}{*}{$(1,0)$} 
 & \multirow{2}{*}{$  \C^5_{\ol{\del}}$} 
 & \multirow{2}{*}{$\C$} 
 & D9-branes on $\C^5$ 
 & \\
&&&& & in IIB[$\C^5_B$] & \\
\hline 
\multirow{2}{*}{9}  
 & \multirow{2}{*}{1}           
 & \multirow{2}{*}{1}                
 & \multirow{2}{*}{$\R_{\mr{dR}} \times \C^4_{\ol{\del}}$} 
 & \multirow{2}{*}{$\C$} 
 & D8-branes on $\R\times \C^4$
 & \\
&&&&& in IIA[$\R^2_A\times \C^4_B$] & \\
\hline
\multirow{6}{*}{8}  
 & \multirow{6}{*}{1}     
 & $(1,0)$  
 & \multirow{2}{*}{$ \C^4_{\ol{\del}}$}  
 & \multirow{2}{*}{$\C[\eps]$} 
 & D7-branes on $\C^4$  
 & \\
&&  pure &&& in IIB[$\C^5_B$] & \\
\cline{3-7}
 & &  \multirow{2}{*}{$(1,1)$}     
 & \multirow{2}{*}{$\R^2_{\mr{dR}} \times \C^3_{\ol{\del}}$}  
 & \multirow{2}{*}{$\C$}   
 & D7-branes on $\R^2 \times \C^3$
 & \\
&&&&&  in IIB[$\R^4_A \times \C^3_B$] & \\
\cline{3-7}
 & &  $(1,0)$  
 & \multirow{2}{*}{$ \C^4_{\ol{\del}}$}  
 & \multirow{2}{*}{$\left (\C[\eps], \frac{\del}{\del \eps } \right) $} 
 & D7-branes on $\C^4$  
 &  $\C^4_{z_i}\subset \C^5_{z_i,w}$ \\
&& impure &&& in IIB[$\C^5_B$] & $w\in \PV(\C^5)$ \\
\hline
\multirow{6}{*}{7} 
 & \multirow{6}{*}{1} 
 & 1 
 & \multirow{2}{*}{$\R_{\mr{dR}} \times \C^3_{\ol{\del}}$} 
 & \multirow{2}{*}{$\C[\eps]$}  
 & D6-branes on $\R \times \C^3$
 & \\
&& pure &&& in IIA[$\R^2_A\times \C^4_B$] & \\
\cline{3-7}
&& \multirow{2}{*}{2} 
 & \multirow{2}{*}{$\R^3_{\mr{dR}} \times \C^2_{\ol{\del}}$}  
 & \multirow{2}{*}{$\C$}  
 & D6-branes on $\R^3\times \C^2$
 & \\
&&&& & in IIA[$\R^6_A\times \C^2_B$] & \\
\cline{3-7}
&& 1
 & \multirow{2}{*}{$\R_{\mr{dR}} \times \C^3_{\ol{\del}}$} 
 & \multirow{2}{*}{$\left (\C[\eps], \frac{\del}{\del \eps } \right)$}  
 & D6-branes on $\R\times \C^3$
 & $\C^3_{z_i}\subset \C^4_{z_i,w}$  \\
&& impure &&& in IIA[$\R^2_A\times \C^4_B$] &  $w\in \PV(\C^4)$   \\
\hline
\multirow{8}{*}{6} 
 & \multirow{8}{*}{$(1,1)$}    
 & \multirow{2}{*}{$(1,0)$}  
 & \multirow{2}{*}{$ \C^3_{\ol{\del}}$} 
 & \multirow{2}{*}{$\C[\eps_1,\eps_2]$}  
 & D5-branes on $\C^3$  
 & \\
 & & & & & in IIB[$\C^5_B$] & \\
\cline{3-7}
 & &   $(1,1)$  
 & \multirow{2}{*}{$\R^2_{\mr{dR}} \times \C^2_{\ol{\del}}$}  
 & \multirow{2}{*}{$\C[\eps]$}  
 & D5-branes on $\R^2\times \C^2$
 & \\
&& special &&& in IIB[$\R^4_A\times\C^3_B$] & \\
\cline{3-7}
&& \multirow{2}{*}{$(2,2)$}
 & \multirow{2}{*}{$\R^4_{\mr{dR}} \times \C_{\ol{\del}}$}  
 & \multirow{2}{*}{$\C$}  
 & D5-branes on $\R^4\times \C$
 & \\
&&&&& in IIB[$\R^8_A\times \C_B$] & \\
\cline{3-7}
&& $(1,1)$
 & \multirow{2}{*}{$\R^2_{\mr{dR}} \times \C^2_{\ol{\del}}$}  
 & \multirow{2}{*}{$\left (\C[\eps], \frac{\del}{\del \eps }  \right)$}  
 & D5-branes on $\R^2\times \C^2$
 &  $\C^2_{z_i}\subset \C^3_{z_i,w}$  \\
&& generic &&& in IIB[$\R^4_A\times\C^3_B$] & $w\in \PV(\C^3)$ \\
\hline
\multirow{8}{*}{5} 
 & \multirow{8}{*}{2} 
 &  \multirow{2}{*}{1} 
 & \multirow{2}{*}{$\R_{\mr{dR}} \times \C^2_{\ol{\del}}$}  
 & \multirow{2}{*}{$\C[\eps_1,\eps_2]$} 
 & D4-branes on $\R \times \C^2$
 & \\
&&&&& in IIA[$\R^2_A\times \C^4_B$] & \\
\cline{3-7}
&& 2
 & \multirow{2}{*}{$\R^3_{\mr{dR}} \times \C_{\ol{\del}}$} 
 & \multirow{2}{*}{$\C[\eps]$} 
 & D4-branes on $\R^3\times \C$
 & \\
&& special &&& in IIA[$\R^6_A\times \C^2_B$] & \\
\cline{3-7}
&& \multirow{2}{*}{4}
 & \multirow{2}{*}{$\R^5_{\mr{dR}}$}  
 & \multirow{2}{*}{$\C$}  
 & D4-branes on $\R^5$
 & \\
&&&&& in IIA[$\R^{10}_A$] & \\
\cline{3-7}
&& 2
 & \multirow{2}{*}{$\R^3_{\mr{dR}} \times \C_{\ol{\del}}$} 
 & \multirow{2}{*}{$\left (\C[\eps], \frac{\del}{\del \eps }\right )$} 
 & D4-branes on $\R^3\times \C$
 & $\C_{z}\subset \C^2_{z,w}$ \\
&& generic &&& in IIA[$\R^6_A\times \C^2_B$] & $w\in \PV(\C^2)$  \\
\hline
\multirow{12}{*}{4}
 & \multirow{12}{*}{4} 
 &  \multirow{2}{*}{$(1,0)$} 
 & \multirow{2}{*}{$\C^2_{\mr{Dol}}$} 
 & \multirow{2}{*}{$\C[\eps]$}   
 & D3-branes on $\C^2$
 & \\
&&&&& in IIB[$\C^5_B$] \\
\cline{3-7}
&& \multirow{2}{*}{$(1,1)$}
 & \multirow{2}{*}{$\R^2_{\mr{dR}} \times \C_{\mr{Dol}}$} 
 & \multirow{2}{*}{$\C[\eps]$}   
 & D3-branes on $\R^2 \times \C$
 & \\
&&&&& in IIB[$\R^4_A\times\C^3_B$] \\
\cline{3-7}
&& $(2,2)$
 & \multirow{2}{*}{$\R^4_{\mr{dR}}$}  
 & \multirow{2}{*}{$\C[\eps]$} 
 & D3-branes on $\R^4$
 & \\
&& special &&& in IIB[$\R^8_A\times\C_B$] \\
\cline{3-7}
&&\multirow{2}{*}{$(2,0)$}
 & \multirow{2}{*}{$\C^2_{\mr{Dol}}$}  
 & \multirow{2}{*}{$\left (\C[\eps], \frac{\del}{\del \eps }\right )$}  
 & D3-branes on $\C^2$
 &  $\C^2_{z_i}\subset \C^5_{z_i,w_j,w}$ \\
&&&&& in IIB[$\C^5_B$] & $w \in \PV(\C^5)$ \\
\cline{3-7}
&&\multirow{2}{*}{$(2,1)$}
 &  \multirow{2}{*}{$\R^2_{\mr{dR}} \times \C_{\mr{Dol}}$} 
 & \multirow{2}{*}{$\left (\C[\eps], \frac{\del}{\del \eps }\right )$} 
 & D3-branes on $\R^2 \times \C$
 & $\C_{z_1}\subset \C^3_{z_1,w_1,w}$ \\
&&&&& in IIB[$\R^4_A\times\C^3_B$] &  $w \in \PV(\C^3)$  \\
\cline{3-7}
&& $(2,2)$
 & \multirow{2}{*}{$\R^4_{\mr{dR}}$} 
 & \multirow{2}{*}{$\left (\C[\eps], \frac{\del}{\del \eps }\right )$}
 & D3-branes on $\R^4$
 & $\{0\}\subset \C_w$  \\
&& generic &&& in IIB[$\R^8_A\times\C_B$]& $w\in \PV(\C)$ \\
\hline
\multirow{6}{*}{3}  
 & \multirow{6}{*}{8}
 & \multirow{2}{*}{1} 
 & \multirow{2}{*}{$\R_{\mr{dR}} \times \C_{\mr{Dol}}$}
 & \multirow{2}{*}{$\C[\eps_1,\eps_2]$} 
 & D2-branes on $\R \times \C$
 & \\
&&&&& in IIA[$\R^2_A\times\C^4_B$] \\
\cline{3-7}
&&\multirow{2}{*}{2 (B)}
 & \multirow{2}{*}{$\R^3_{\mr{dR}}$} 
 & \multirow{2}{*}{$\C[\eps_1,\eps_2]$} 
 & D2-branes on $\R^3$
 & \\
&&&&& in IIA[$\R^6_A\times\C^2_B$] \\
\cline{3-7}
&&\multirow{2}{*}{2 (A)}
 & \multirow{2}{*}{$\R^3_{\mr{dR}}$} 
 & \multirow{2}{*}{$\left (\C[\eps_1,\eps_2], \frac{\del}{\del \eps_1}\right )$} 
 & D2-branes on $\R^3$
 & $\{0\}\subset \C^2_{w_1,w_2}$ \\
&&&&& in IIA[$\R^6_A\times\C^2_B$] & $w_1\in \PV(\C^2)$  \\
\hline
\end{tabular}
\caption{Twists of maximally supersymmetric Yang--Mills theories as open-string field theories}\label{table1}
\end{table}

\begin{table}
\centering
\renewcommand{\arraystretch}{0.6}
\begin{tabular}{c|c|c||c|c||c|c}
{\(d\)} 
 & \( {\mathcal{N}}\)
 & \text{twist}
 & \text{spacetime}
 & \( {A}\)
 & \text{D-branes}
 & \text{closed string} \\
\hline
\multirow{2}{*}{6}  & \multirow{2}{*}{$(1,0)$} & \multirow{2}{*}{$(1,0)$} & \multirow{2}{*}{$\C^3_{\ol{\del}}$} & \multirow{2}{*}{$\C[\delta]/(\delta^2)$} & D7-branes on $\C^3 \times \bb P^1 $ in  \\
&&&&& IIB[$(\C^3 \times \tot_{\bb P^1}(\mc O(-2)))_B$]  \\
\hline
\multirow{2}{*}{5}  & \multirow{2}{*}{1} & \multirow{2}{*}{1}  & \multirow{2}{*}{$ \R_{\mr{dR}} \times \C^2_{\ol{\del}} $}  & \multirow{2}{*}{$\C[\delta]/(\delta^2)$}  & D6-branes on $\R \times \C^2 \times  \bb P^1 $ in \\
&&&&& IIA[$\R^2_A \times (\C^2 \times \tot_{\bb P^1}(\mc O(-2)))_B$] \\
\hline
\multirow{6}{*}{4} & \multirow{6}{*}{2}& \multirow{2}{*}{$(1,0)$} & \multirow{2}{*}{$\C^2_{\ol{\del}}$} & \multirow{2}{*}{$\C[\eps,\delta]/(\delta^2)$} &  D5-branes on $\C^2 \times \bb P^1 $ in   &   \\
&&&&& IIB[$(\C^3 \times \tot_{\bb P^1}(\mc O(-2)))_B$]  &   \\
\cline{3-7}
&&\multirow{2}{*}{$(1,1)$} &\multirow{2}{*}{$\R^2_{\mr{dR}}\times \C_{\ol{\del}}$}  & \multirow{2}{*}{$\C[\delta]/(\delta^2)$}  & D5-branes on $\R^2\times  \C \times \bb P^1 $  in \\
&&&&& IIB[$\R^4_A\times  (\C \times \tot_{\bb P^1}(\mc O(-2)))_B$]  \\
\cline{3-7}
&&\multirow{2}{*}{$(2,0)$}& \multirow{2}{*}{$\C^2_{\ol{\del}} $}  &\multirow{2}{*}{$\left (\C[\eps,\delta]/(\delta^2) , \frac{\del}{\del \eps}\right ) $} &  D5-branes on $\C^2 \times \bb P^1 $ in    & $\C^2_{z_i}\subset \C^3_{z_i,w}$ \\
&&&&& IIB[$(\C^3 \times \tot_{\bb P^1}(\mc O(-2)))_B$] & $w\in \PV(\C^3)$  \\
\hline
\multirow{6}{*}{3} & \multirow{6}{*}{4} & \multirow{2}{*}{1} & \multirow{2}{*}{$ \R_{\mr{dR}}\times \C_{\ol{\del}} $} & \multirow{2}{*}{$\C[\eps,\delta]/(\delta^2)$} & D4-branes on $\R\times \C\times \bb P^1$ in &  \\
&&&&&IIA[$\R^2_A \times (\C^2 \times \tot_{\bb P^1}(\mc O(-2)))_B$] &  \\
\cline{3-7}
&&	 \multirow{2}{*}{2 (B)}  & \multirow{2}{*}{$\R^3_{\mr{dR}}$}  & \multirow{2}{*}{$\C[\delta]/(\delta^2)$} & D4-branes on $\R^3 \times \bb P^1$ in \\
&&&&& IIA[$\R^6_A \times \tot_{\bb P^1}(\mc O(-2))_B$] \\
\cline{3-7}
&&  \multirow{2}{*}{2 (A)}& \multirow{2}{*}{$ \R_{\mr{dR}}\times \C_{\ol{\del}}$} & \multirow{2}{*}{$\left (\C[\eps,\delta]/(\delta^2) , \frac{\del}{\del \eps}\right ) $} &  D4-branes on $\R\times \C\times \bb P^1$ in & $\C_z\subset \C^2_{z,w}$ \\
&&&&&  IIA[$\R^2_A \times (\C^2 \times \tot_{\bb P^1}(\mc O(-2)))_B$] & $w\in \PV(\C^2)$ \\
\hline
\end{tabular}
\caption{Twists of pure supersymmetric Yang--Mills theories with 8 supercharges as open-string field theories}\label{table2}
\end{table}

\begin{table}
\centering
\renewcommand{\arraystretch}{0.6}
\begin{tabular}{c|c|c||c|c||c}
{\(d\)} 
 & \( {\mathcal{N}}\)
 & \text{twist}
 & \text{spacetime}
 & \( {A}\)
 & \text{D-branes}\\
\hline
\multirow{4}{*}{4}  & \multirow{4}{*}{1} & \multirow{4}{*}{$(1,0)$} & \multirow{4}{*}{$\C^2_{\ol{\del}}$} & \multirow{4}{*}{$\C[\eps']$}  & D5-branes on $\C^2 \times \bb P^1 $ in  \\
&&&&&  IIB[$(\C^2 \times \tot_{\bb P^1}( \mc O(-1)\oplus \mc O(-1))) _B$]   \\
\cline{6-6}
 & &   &  &   & D7-branes on $\C^2 \times  \bb P^2$ in  \\
&&&&& IIB[$(\C^2 \times \tot_{\bb P^2}(\mc O(-3)) )_B$]   \\
\hline
\multirow{4}{*}{3} & \multirow{4}{*}{2}& \multirow{4}{*}{1} & \multirow{4}{*}{$ \R_{\mr{dR}}\times \C_{\ol{\del}}$} & \multirow{4}{*}{$ \C[\eps'] $} & D4-branes on $\R \times \C \times \bb P^1 $ in   \\
&&&&&  IIA[$\R^2_A\times  (\C \times \tot_{\bb P^1}( \mc O(-1)\oplus \mc O(-1))) _B$]   \\
\cline{6-6}
 & &   &  &   & D6-branes on $\R \times  \C \times  \bb P^2$ in  \\
&&&&& IIA[$\R^2_A \times  (\C \times \tot_{\bb P^2}(\mc O(-3)) )_B$]   \\
\hline 
\end{tabular}
\caption{Twists of pure supersymmetric Yang--Mills theories with 4 supercharges as open-string field theories}\label{table3}
\end{table}

\begin{table}
\centering
\renewcommand{\arraystretch}{0.6}
\begin{tabular}{c|c||c|c||c|c}
 \( {\mathcal{N}}\) & \text{twist} & \text{spacetime} & \( {A}\) & \text{D-branes} & \text{closed string} \\
\hline
 \multirow{6}{*}{$(4,4)$} & \multirow{2}{*}{$(1,0)$} & \multirow{2}{*}{$\C_{\ol{\del}}$} & \multirow{2}{*}{$\C[\eps_1,\eps_2,\delta]/(\delta^2 )$} & D3-branes on $\C \times \bb P^1$  in \\
&&&& IIB[$(\C^3 \times \tot_{\bb P^1}(\mc O(-2)))_B$]  \\
\cline{2-6}
& \multirow{2}{*}{$(1,1)$ (B)} & \multirow{2}{*}{$\R^2 _{\mr{dR}}$}   & \multirow{2}{*}{$\C[\eps,\delta]/(\delta^2 )$} &  D3-branes on $\R^2 \times \bb P^1$  in \\
&&&& IIB[$\R^4_A \times (\C \times \tot_{\bb P^1}(\mc O(-2)))_B$]  \\
\cline{2-6}
&  \multirow{2}{*}{$(1,1)$ (A)}&\multirow{2}{*}{$\R^2 _{\mr{dR}}$}  & \multirow{2}{*}{$\left( \C[\eps,\delta]/(\delta^2 ) ,\frac{\del}{\del \eps } \right)$} &  D3-branes on $\R^2 \times \bb P^1$  in   &  $\{0\}\subset \C_w$ \\
&&&& IIB[$\R^4_A \times (\C \times \tot_{\bb P^1}(\mc O(-2)))_B$]  & $w\in \PV(\C)$ \\
\hline
 \multirow{12}{*}{$(2,2)$} & \multirow{4}{*}{$(1,0)$} & \multirow{4}{*}{$\C_{\ol{\del}}$} & \multirow{4}{*}{$\C[\eps,\eps']$} & D3-branes on $\C \times \bb P^1$  in \\
&&&& IIB[$(\C^2 \times \tot_{\bb P^1}(\mc O(-1)\oplus\mc O(-1) ))_B$]  \\
\cline{5-5}
&&&& D5-branes on $\C \times \bb P^2$  in \\
&&&& IIB[$(\C^2 \times \tot_{\bb P^2}(\mc O(-3)) )_B$]  \\
\cline{2-6}
& \multirow{4}{*}{$(1,1)$ (B)} & \multirow{4}{*}{$\R^2 _{\mr{dR}}$}   & \multirow{4}{*}{$\C[\eps']$} &  D3-branes on $\R^2 \times \bb P^1$  in \\
&&&& IIB[$\R^4_A \times  \tot_{\bb P^1}(\mc O(-1)\oplus\mc O(-1) )_B$]  \\
\cline{5-5}
&&&&  D5-branes on $\R^2 \times \bb P^2$  in \\
&&&& IIB[$\R^4_A \times  \tot_{\bb P^2}(\mc O(-3))_B$]  \\
\cline{2-6}
&  \multirow{4}{*}{$(1,1)$ (A)}&\multirow{4}{*}{$\C_{\ol{\del}}$}  & \multirow{4}{*}{$\left(\C[\eps,\eps'] ,\frac{\del}{\del \eps } \right)$} &  D3-branes on $\C \times \bb P^1$  in   &    \\
&&&& IIB[$(\C^2 \times \tot_{\bb P^1}(\mc O(-1)\oplus\mc O(-1) ))_B$] & $\C_z\subset \C^2_{z,w}$ \\
\cline{5-5}
&&&&  D5-branes on $\C \times \bb P^2$  in   &    $w\in \PV(\C^2)$ \\
&&&& IIB[$(\C^2 \times \tot_{\bb P^2}(\mc O(-3) ))_B$]  \\
\hline  \multirow{2}{*}{$(4,0)$} &  \multirow{2}{*}{$(1,0)$} &   \multirow{2}{*}{\(\C_{\ol{\del}}\) }&  \multirow{2}{*}{\(\C[\delta_1,\delta_2]/(\delta_1^2,\delta_2^2)\)}  & D5-branes on $\C\times \bb P^1 \times \bb P^1$ in  \\  
&&&& $\IIB[(\C \times \tot_{\bb P^1 \times \bb P^1}( \Omega^1_{\bb P^1 \times \bb P^1} ))_B ]$ \\
\hline 
\multirow{4}{*}{$(2,0)$} &  \multirow{4}{*}{$(1,0)$} &   \multirow{4}{*}{\(\C_{\ol{\del}}\) }&  \multirow{4}{*}{\(\C[\delta]/(\delta^2)\)}  & D5-branes on $\C\times \bb P^2$ in  \\  
&&&& IIB[$(\C \times \tot_{\bb P^2}(\mc O(-1) \oplus \mc O(-2)   ) )_B$]  \\
\cline{5-5}
 &  &  &  & D7-branes on $\C\times \bb P^3$ in  \\  
&&&& IIB[$(\C \times \tot_{\bb P^3}(\mc O(-4)   ) )_B$]    \\
\hline 
\end{tabular}
\caption{Twists of pure supersymmetric Yang--Mills theories in two dimensions as open-string field theories}\label{table4}
\end{table}

\FloatBarrier

The main body of this article is devoted to a precise formulation and proof of this theorem, along with several related remarks. We also explain how to interpret the table that summarizes key results. Once the key mathematical structures are clarified, the argument becomes straightforward.

In Section 2, we review the classical BV (Batalin--Vilkovisky) formalism and introduce a class of Chern--Simons-like theories constructed from cyclic graded-commutative algebras. The main result of this section, Proposition \ref{prop:twistedSYMisCS}, asserts that every twist of pure supersymmetric Yang--Mills theory is such a generalized Chern--Simons theory. In Section 3, we present another main result, Proposition \ref{prop:CSfromTS}, which shows how each such generalized Chern--Simons theory arises as an open-string field theory from a suitable topological string background. The main theorem follows immediately from these two propositions. After proving it, we provide further discussion and remarks.

\subsubsection*{Convention}: Throughout this article, we work with a complexified gauge group, specifically $\GL(N)$ instead of $\U(N)$. We also use the classical BV formalism in a $\Z/2\Z$-graded setting, while presenting everything in a form that can be interpreted as $\Z$-graded whenever possible.

\subsubsection*{Acknowledgements}: We are grateful to Surya Raghavendran for valuable discussions on various aspects of string theory and for providing detailed comments on an initial draft of this article, which significantly improved the exposition. This work was supported by the National Research Foundation of Korea (NRF) through grants funded by the Korean government (No. 2022R1F1A107114213) and the Ministry of Education under the LAMP Program (No. RS-2023-00301976).

\section{Generalized Chern--Simons Theories}

\subsection{Classical BV Formalism and AKSZ Formalism}

In this subsection, we introduce the classical BV formalism as a framework for describing classical field theory. After discussing the AKSZ formalism, we introduce a class of field theories of our interest, called generalized Chern--Simons theories. Since our contribution makes use of the framework in a way that does not depend on detailed formal definitions, we provide only a brief overview with illustrative examples. Readers are encouraged to refer to the cited papers for additional details and further context.

\subsubsection{} 

We begin by recalling the definition of a perturbative classical field theory in the BV formalism, following \cite{CostelloSUSY}.

\begin{defn}\label{defn:pcft}
Let $M$ be a $d$-dimensional manifold. A \emph{($d$-dimensional) perturbative classical field theory} on $M$ is a local formal moduli problem on $M$ with a $(-1)$-shifted symplectic structure, or equivalently, a local $L_\infty$-algebra on $M$ with a non-degenerate invariant pairing of cohomological degree $-3$. 
\end{defn}

\begin{rmk}\label{rmk:ftdt}
This definition is meant to describe the structure present on the formal neighborhood of a point in the moduli space of solutions to a variational PDE; see Remark \ref{rmk:triangle} as well. 

Here $M$ denotes a spacetime, where locality --- a fundamental requirement of field theory --- is defined. The existence of a $(-1)$-shifted symplectic structure on the moduli space of solutions to the equations of motion is a key feature of the classical BV formalism. Furthermore, the fundamental theorem of deformation theory by Pridham and Lurie \cite{Pridham, LurieModuli} establishes an equivalence
\[\xymatrix{
\text{Moduli}\ar@<0.5ex>[rr]^-{\Omega= \bb{T}[-1]} &  & L_\infty\text{-Alg}\ar@<0.5ex>[ll]^-{B = \mr{MC}}
}\]
where $\text{Moduli}$ is the category of formal moduli problems, $L_\infty\text{-Alg}$ is the category of $L_\infty$-algebras, $\Omega =\bb T[-1]\colon \text{Moduli}\to L_\infty\text{-Alg}$ is the shifted tangent complex functor, and $B=\mr{MC} \colon L_\infty\text{-Alg}\to \text{Moduli}$ is the Maurer--Cartan functor. Under this equivalence, a $k$-shifted symplectic structure on a formal moduli problem $F$ corresponds to a non-degenerate invariant pairing of degree $k-2$ on the $L_\infty$-algebra $\bb T[-1]F$. 
\end{rmk}

\begin{ex}[Perturbative Chern--Simons theory]\label{ex:perturbativeCS}
Let $M$ be a (compact oriented) 3-manifold. Let $G$ be a semisimple Lie group with a fixed symmetric invariant pairing $\langle -,-\rangle_{\mf g }$ on the Lie algebra $\mf g$.

Chern--Simons theory is traditionally described with the space of fields $\Omega^1(M, \mf g)$, the $1$-forms on $M$ valued in the Lie algebra $\mf g$, interpreted as the space of connections on the trivial principal $G$-bundle over $M$. The Chern--Simons action functional $S_{\mr{CS}}$ is given by \[S_{\mr{CS}} (A) = \frac{1}{2}\int_M \langle A, d A \rangle  + \frac{1}{6}\int_M \langle A,[A,A]\rangle  \]
where $\langle -,-\rangle$ is the pairing induced from the wedge product of differential forms and the invariant pairing $\langle -,-\rangle_{\mf g}$ on $\mf g$. The theory has infinitesimal gauge symmetries represented by $\Omega^0(M, \mf g)$, since the action functional is invariant under the transformation $A \mapsto dX + [X, A]$ for $X \in \Omega^0(M, \mf g)$. To derive the equations of motion, we vary with respect to $A_0 \in \Omega^1_c(M,\mf g) $, which yields $\frac{\delta S_{\mr{CS}} }{\delta A_0}(A) = \int_M \langle A_0, F_A\rangle$, where $F_A = dA + \frac{1}{2}[A,A]$ is the curvature of the connection $A$. Hence the space of solutions to the equations of motion consists of flat $G$-connections (modulo gauge transformations).

We now describe how Chern--Simons theory is formulated within the BV framework, in terms of a local DG Lie algebra\footnote{The locality ensures that the following works over any open subset $U$ of $M$.} with an invariant pairing:
\begin{itemize}
	\item the graded vector space in degrees 0, 1, 2, and 3 \[\xymatrix@R-20pt{
\ul{0} & \ul{1} &\ul{2} &\ul{3} \\
\Omega^0(M,\mf g)   & \Omega^1(M,\mf g)   & \Omega^2(M,\mf g )    &  \Omega^3(M ,\mf g)  }\] is the space of global sections of a graded vector bundle $\bigoplus_{i=0}^3 \wedge^i(T^*M)\otimes \mf g$ over $M$;
	\item the differential $d \colon \Omega^i(M,\mf g)\to \Omega^{i+1}(M,\mf g)$ is induced from the de Rham differential on $ \Omega^\bullet(M)$;
	\item the graded Lie bracket $[-,-]\colon \Omega^i(M,\mf g)\otimes \Omega^j(M,\mf g)\to \Omega^{i+j}(M,\mf g)$ is induced from the graded-commutative algebra structure on $\Omega^\bullet(M)$ and the Lie bracket on $\mf g$;
	\item the symmetric invariant pairing $\abracket{-,-}\colon \Omega^\bullet (M,\mf g)\otimes \Omega^\bullet(M,\mf g)\to  \Omega^3(M)[-3]$ is induced by the algebra structure on $\Omega^\bullet(M)$ and the invariant pairing $\langle -,-\rangle_{\mf g }$ on $\mf g$. Here, $\Omega^3(M)$ denotes the space of the top forms, i.e., sections of the density line bundle on $M$.
\end{itemize}

From this description of a local DG Lie algebra, one can recover the traditional description in terms of an action functional. That is, consider the BV action functional \[S (\alpha) := \frac{1}{2}\int_M \langle \alpha  , d\alpha\rangle + \frac{1}{6} \int_M \langle \alpha,[\alpha,\alpha]\rangle,\]
where $\alpha \in \mc E_{\mr{CS}} := \Omega^\bullet(M, \mf g)[1]$ represents a general field in the BV formalism. If we write $\alpha=X+A+A^\vee+X^\vee$, where $X\in \Omega^0(M,\mf g)$, $A \in \Omega^1(M,\mf g)$, $A^\vee \in \Omega^2(M,\mf g)$, and $X^\vee \in \Omega^3(M,\mf g)$, then the BV action functional expands out to yield \[S  (X,A,A^\vee,X^\vee)  = S _{\mr{CS}} (A) + \int_M  \abracket{  A^\vee , dX + [X,A]} + \frac{1}{2} \int_M \abracket{   X^\vee, [X,X] } .\] Here the term involving $A^\vee$ encodes gauge transformations. In fact, the space of solutions to the Maurer--Cartan equations of this DG Lie algebra recovers flat $G$-connections modulo gauge transformations.  
\end{ex}

\begin{rmk}
In general, a perturbative classical field theory in the Lagrangian formalism can be defined in terms of the BV space of fields $\mc E$ with a $(-1)$-shifted symplectic form $\omega$ and a local $L_\infty$-algebra structure $(Q, \ell_2, \ell_3, \dots)$ on its cohomological shift $\mc E[-1]$. The corresponding action functional $S$ can then be expressed as
\[ S(\alpha) = \frac{1}{2} \omega(\alpha, Q \alpha) + \sum_{k \geq 2} \frac{1}{(k+1)!} \omega(\alpha, \ell_k(\alpha, \dots, \alpha)) \]
where $\alpha \in \mc E$. Thus, this definition of perturbative classical field theory recovers and extends the traditional formulation of classical field theory.   
\end{rmk}

\begin{rmk}\label{rmk:cptori}
Neither $M$ being compact nor oriented is necessary in this framework. Without these assumptions, we still have a $(-1)$-shifted symplectic structure with respect to Verdier duality, where the invariant pairing takes values in the density line bundle rather than in the space of top forms.
\end{rmk}

\subsubsection{} We now provide an overview of non-perturbative classical field theory. While this material is not strictly necessary for what follows, it offers an important perspective on the underlying framework.

As mentioned in Remark \ref{rmk:ftdt}, a characterizing feature of the classical BV formalism is the existence of a $(-1)$-shifted symplectic structure on the spaces of solutions to the equations of motion. Hence, a $d$-dimensional non-perturbative classical field theory on $M$ can be described as a sheaf of derived stacks equipped with a $(-1)$-shifted symplectic structure (in the sense of Verdier) on $M$. A rigorous definition of non-perturbative classical field theory along these lines has not yet appeared in the literature. On the other hand, we have examples of a non-perturbative classical field theory from AKSZ--PTVV formalism introduced in \cite{PTVV}:
\begin{thm}\label{thm:PTVV}
Let $\mc Y,\mc Z$ be derived stacks. Suppose $\mc Y$ is equipped with an $m$-orientation and $\mc Z$ is equipped with an $n$-shifted symplectic structure. Then the mapping stack $\umap(\mc Y,\mc Z)$ carries an induced $(n-m)$-shifted symplectic structure.  
\end{thm}

When $\mc Y$ is constructed from a $d$-dimensional real manifold, and $\umap(\mc Y, \mc Z)$ is a $(-1)$-shifted symplectic derived stack, we interpret it as the non-perturbative data of a $d$-dimensional classical field theory. We then call the resulting theory the \emph{AKSZ (non-perturbative) field theory} described by $\umap(\mc Y, \mc Z)$.

\begin{rmk}\label{rmk:triangle}
We now explain how this result relates to the perturbative framework described earlier, continuing from Remark \ref{rmk:ftdt}. Consider the following commuting triangle
\[\xymatrix@R-10pt{
 && \text{dSt}_\ast  \ar[dll]_{ (-)^\wedge_\ast } \ar[drr]^-{ \bb{T}_\ast [-1]  } \\
\text{Moduli}\ar@<0.5ex>[rrrr]^-{\Omega= \bb{T}[-1]} & && & L_\infty\text{-Alg}\ar@<0.5ex>[llll]^-{B = \mr{MC}}
}\]
where
\begin{itemize}
	\item $\text{dSt}_\ast $ is the category of pointed derived stacks $(\mc X,x)$, where $\mc X$ is a derived stack and $x\in \mc X(\C)$,
	\item the completion functor $(-)^\wedge_\ast \colon \text{dSt}_\ast\to \text{Moduli} $ is given by $(\mc X,x)\mapsto \mc X^\wedge_x$, and
	\item the shifted tangent complex functor $\bb T_\ast [-1] \colon \text{dSt}_\ast\to L_\infty\text{-Alg}$ is given by $(\mc X,x)\mapsto \bb T_x [-1] \mc X$.
	\end{itemize}
The completion functor $(-)^\wedge_\ast$ corresponds to the formal neighborhood around a given point on the derived stack, while the shifted tangent complex $\bb T_\ast[-1]$ encodes the infinitesimal deformation theory at that point. In a physical situation, a derived stack $\mc X$ arises as the moduli space of solutions to the equations of motion. For a particular solution $x\in \mc X$, the resulting $L_\infty$-algebra structure on $\bb T_x[-1]\mc X$ is precisely the $L_\infty$-algebra structure of Definition \ref{defn:pcft}. This explains how perturbative field theory emerges as an infinitesimal approximation around a specific solution within the broader non-perturbative framework.
\end{rmk}

\begin{ex}[Non-perturbative Chern--Simons theory]
A closed oriented $d$-manifold $M$ gives a $d$-oriented derived stack $M_{\text{B}}$ called a \emph{Betti stack}. A reductive group $G$ gives a 2-shifted symplectic stack $ BG$ whose symplectic structure is induced from the invariant pairing $\abracket{-,-}_{\mf g}$. Then (3-dimensional) non-perturbative Chern--Simons theory is described by a $(-1)$-shifted symplectic stack
 \[\Loc_G(M): = \umap(M_{\mr{B}},BG)\]
which represents the moduli space of $G$-local systems on $M$. To relate this to the perturbative framework, we examine the trivial local system on $M$. The shifted tangent complex $\bb{T}_{\text{triv}}[-1] \Loc_G(M)$, which encodes the infinitesimal deformations around the trivial local system, is given by the DG Lie algebra $( \Omega^\bullet(M , \mf g) , d  , [-,-])$ together with an invariant pairing $\abracket{-,-}$, recovering the description in Example \ref{ex:perturbativeCS}. In other words, the perturbative description of Chern--Simons theory arises from probing the formal neighborhood around the trivial local system in the non-perturbative moduli space.
  
\end{ex}

\subsubsection{} We now discuss how perturbative classical field theory arises from non-perturbative field theory described by mapping stacks $\umap(\mc Y, \mc Z)$. A crucial point is that if $\mc Y$ satisfies suitable finiteness properties (see below for examples), the tangent complex of the mapping stack $\umap(\mc Y, \mc Z)$ is given by $\bb{T}_{f} \umap(\mc Y, \mc Z) \simeq \R \Gamma(\mc Y, f^* \bb{T}_{\mc Z})$.  Hence, an understanding of $\mc O_{\mc Y}$ and $\Gamma(\mc Y, \mc O_{\mc Y})$ is essential for extracting a perturbative description from a given non-perturbative description. In fact, an \emph{$m$-orientation} on $\mc Y$ is defined by a ``fundamental class'' on $\mc Y$, which is a linear map $\int_{[\mc Y]} \colon \Gamma(\mc Y,\mc O_{\mc Y} )\to \C[-m]$.

\begin{ex}\hfill \label{ex:orient}
\begin{itemize}
	\item For a compact, smooth, oriented $d$-manifold $M$, the de Rham space \[M_{\mr{dR}} := ( M, (\Omega^\bullet_M,d_M)  )\] is $d$-oriented, with the fundamental class \[\int_{[M]} \colon \Gamma(M,  (\Omega^\bullet_M,d_M)) =H_{\mr{dR}}^\bullet(M) \to \C[-d].\]
	\item For a smooth, proper complex $d$-manifold $X$, the Dolbeault space \[ X_{\mr{Dol}} : = ( X, (\Omega^{\bullet ,\bullet}_X , \ol{\del}_X )  )\] is $2d$-oriented, with the fundamental class given by \[\int_{[X]} \colon \Gamma(X,  (\Omega^{\bullet,\bullet}_X, \ol{\del}_X )) =H_{\mr{Dol}}^{\bullet,\bullet} (X) \to \C[-2d].\]
	\item For a smooth, proper complex $d$-manifold $X$, the space \[X_{\ol{\del}} := ( X, (\Omega^{0,\bullet}_X , \ol{\del}_X ) )\] has no canonical orientation because  $\Gamma( X, (\Omega^{0,\bullet}_X , \ol{\del}_X ) ) =H_{\mr{Dol}}^{0,\bullet }(X) $. However, if $X$ is equipped with a Calabi--Yau structure $\Omega_X$, then  we have  the fundamental class given by \[\int_{[X]} \colon \Gamma( X, (\Omega^{0,\bullet}_X , \ol{\del}_X ) ) = H_{\mr{Dol}}^{0,\bullet }(X) \cong H_{\mr{Dol}}^{d,\bullet }(X) \to \C[-d].\] In other words, for a Calabi--Yau $d$-fold $X$, the space   $X_{\ol{\del}  }$ is canonically $d$-oriented. 
\end{itemize}	
\end{ex}

\begin{rmk}
As noted in Remark \ref{rmk:cptori}, compactness or properness is necessary to obtain orientation data in a strict sense, but these conditions are not necessary for the BV description.	 
\end{rmk}

\begin{ex}
An important example is perturbative Chern--Simons theory described by
\[\umap(M_{\mr{dR}}, B \mf g ) := \Omega^\bullet(M , \mf g)[1],\]
or equivalently, in terms of local $L_\infty$-algebra, $\bracket{ \Omega^\bullet(M , \mf g), d, [-,-], \abracket{-,-} }$. 	
\end{ex}

\begin{rmk}
This raises the question of how global, non-perturbative field theory $\umap(M_{\mr{B}}, BG)$ in an algebraic context relates to formal, perturbative field theory $\umap(M_{\mr{dR}}, B \mf g)$ in a smooth context. For further discussion on this point, see \cite[Remark 2.5]{BY1}.   
\end{rmk}

\subsubsection{} According to Theorem \ref{thm:PTVV} and Example \ref{ex:orient}, the mapping stack 
\[ \umap( M_{\mr{dR}} \times  X_{\ol{\del}}\times Y_{\mr{Dol}} , B\mf g ):=\Omega^\bullet(M)\otimes \Omega^{0,\bullet}(X)  \otimes \Omega^{\bullet,\bullet}(Y)\otimes  \mf g[1]  \]
is $(-1)$-shifted symplectic if $\dim _\C X + 2 \dim_\C Y  + \dim_\R M=3$, defining a perturbative classical field theory in the sense of Definition \ref{defn:pcft}. This is by definition \emph{AKSZ (perturbative) field theory} described by $\umap( M_{\mr{dR}} \times  X_{\ol{\del}}\times Y_{\mr{Dol}} , B\mf g )$. We now introduce a generalization of this construction, where $\mf g$ is replaced by any \emph{cyclic $L_\infty$-algebra} --- an $L_\infty$-algebra with a non-degenerate invariant pairing --- which we denote generically by $\mf l$.

\begin{defn}\label{defn:genCS}\cite{ESW}
Let $M$ be a smooth oriented manifold, $(X,\Omega_X)$ a Calabi--Yau manifold, and $Y$ a complex manifold. Let $(\mf l,d_{\mf l} ,\ell_2,\cdots) $ be a cyclic $L_\infty$-algebra with pairing $\langle -,-\rangle$ of cohomological degree $ \dim_\R M+\dim _\C X + 2 \dim_\C Y  -3$. Then the corresponding \emph{generalized Chern--Simons theory} on $M\times X\times Y$ is defined as the $( \dim_\R M+2\dim _\C X + 2 \dim_\C Y )$-dimensional AKSZ theory described by 
\[ \umap( M_{\mr{dR}}\times  X_{\ol{\del}}\times Y_{\mr{Dol}} , B\mf l ):= \Omega^\bullet(M)\otimes  \Omega^{0,\bullet}(X)  \otimes \Omega^{\bullet,\bullet}(Y)\otimes \mf l[1],  \]
or equivalently, by the associated local $L_\infty$-algebra \[\Omega^\bullet(M)\otimes  \Omega^{0,\bullet}(X)  \otimes \Omega^{\bullet,\bullet}(Y)\otimes  \mf l,\] with the $L_\infty$-algebra structure and cyclic pairing induced from the tensor product of the factors. 
\end{defn}

Note that the degree of the invariant pairing is determined so that the mapping stack is $(-1)$-shifted symplectic.\footnote{However,  in the bulk of the paper, the theories are only $\Z/2\Z$-graded so an odd-symplectic structure suffices.} The corresponding action functional for this generalized Chern--Simons theory is given by
\[S(\alpha ) =  \frac{1}{2} \int_{M\times X\times Y  } \langle \alpha \wedge (d_{M}  +  \ol{\del}_X    + \ol{\del}_Y +  d_{\mf l} ) \alpha  \rangle \wedge   \Omega_X +  \sum_{n\geq 2} \frac{1}{(n+1)!} \int_{M\times X\times Y } \langle \alpha , \ell_n (\alpha,\cdots,\alpha  )\rangle \wedge \Omega_X. \]

\begin{ex}[Chern--Simons theories]\label{ex:CStheories}
Let $\mf g$ be the Lie algebra of a semisimple Lie group. Since it has a non-degenerate invariant symmetric pairing of cohomological degree 0, one can define a corresponding generalized Chern--Simons theory on $M\times X\times Y$ where $\dim_\R M+  \dim _\C X + 2 \dim_\C Y  =3$. Note that a field theory satisfying this condition has dimension $d = \dim_\R M+  2\dim _\C X + 2 \dim_\C Y$. The following table lists the possible configurations for generalized Chern--Simons theories with $\mf l=\mf g$:
\begin{center}
\begin{tabular}{|c|c|c|c|c|}
\hline
 & $\dim_\R M$ & $\dim_\C X$ & $\dim_\C Y$ & dimension $d$ \\
\hline
 \circled{1} & 3 & 0 & 0 & 3 \\
 \circled{2} & 1 & 0 & 1 & 3 \\
 \circled{3} & 2 & 1 & 0 & 4 \\
 \circled{4} & 0 & 1 & 1 & 4 \\
 \circled{5} & 1 & 2 & 0 & 5 \\
 \circled{6} & 0 & 3 & 0 & 6 \\
\hline 
\end{tabular}
\end{center}
\begin{itemize}
	\item [ \circled{1}] This recovers 3-dimensional Chern--Simons theory. Non-perturbatively, the phase space (or the moduli space of germs of solutions to the equations of motion along a codimension-1 submanifold $\Sigma$) is the moduli space $\Loc_G(\Sigma)$ of $G$-local systems on $\Sigma $, which depends only the topology of $\Sigma$.
	\item[\circled{2}] This recovers what \cite{BY1} calls 3-dimensional \emph{critical} Chern--Simons theory. It is termed “critical” because it corresponds to Chern--Simons theory at the critical level. Its phase space is the moduli space \(\higgs_G(\Sigma)\) of \(G\)-Higgs bundles on \(\Sigma\), which depends on the complex structure of \(\Sigma\).

\item[\circled{3}] This is 4-dimensional Chern--Simons theory. There are exactly three Calabi--Yau 1-folds: \(\C\), \(\C^\times\), and an elliptic curve \(E\). These theories are closely related to integrable systems with a spectral parameter on \(\C\), \(\C^\times\), or \(E\), and hence to the Yangian, quantum loop group, and elliptic quantum group. See \cite{CostelloYangian, CWY1} and subsequent works for more details.

\item[\circled{4}] This is 4-dimensional ``critical'' Chern--Simons theory. In contrast to ordinary four-dimensional Chern--Simons theory, it does not exhibit any topological dependence on spacetime.

\item[\circled{5}] This is 5-dimensional Chern--Simons theory. A variant of this theory is expected to have deep connections to the quantum double loop group. For a more precise statement and further context, see \cite{CostelloM}.

\item[\circled{6}] This is 6-dimensional Chern--Simons theory, also called \emph{holomorphic Chern--Simons theory}, which is closely related to Donaldson--Thomas (DT) invariants of a Calabi--Yau 3-fold.\footnote{One can also regard DT invariants as related to the universal bulk theory \cite{BY1} of holomorphic Chern--Simons theory, namely, the A-twist of 7-dimensional \(\mc N=1\) gauge theory. For further discussion in the context we use here, see \cite[Remark 4.10]{RY1}, where we outline a physical origin of the GW/DT correspondence.}
\end{itemize}
\end{ex}

\begin{ex}[Generalized BF theories]\label{ex:BF}
Generalized BF theories can be understood as special cases of generalized Chern--Simons theories, where the $L_\infty$-algebra is constructed from a semi-direct product.

Let $\mf h$ be an $L_\infty$-algebra. Let $\mf l = \mf h \ltimes (\mf h^*[d-3])$ be the semi-direct product using the coadjoint action of $\mf h$ on $\mf h^*$, equipped with a natural pairing $\langle -,- \rangle$ induced from the canonical pairing between $\mf h$ and $\mf h^*$. In other words, the space of fields is given by
 \[ \umap(M_{\mr{dR}} \times  X_{\ol{\del}} \times Y_{\mr{Dol}}, T^*[d-1] B\mf h )= \Omega^\bullet(M) \otimes\Omega^{0,\bullet}(X) \otimes \Omega^{\bullet,\bullet}(Y)\otimes  \mf h[1]  \oplus \Omega^\bullet(M)\otimes  \Omega^{0,\bullet}(X) \otimes \Omega^{\bullet,\bullet}(Y)\otimes \mf h^*[d-2].  \]
Note that we can identify $\Omega^{0,\bullet}(X) \cong \Omega^{\dim X ,\bullet}(X) $ due to the existence of a volume form $\Omega_X$ on $X$.

If we write $\mc A\in  \Omega^\bullet(M)\otimes\Omega^{0,\bullet}(X)  \otimes \Omega^{\bullet,\bullet}(Y)\otimes \mf h[1]$ and $\mc B \in \Omega^\bullet(M)\otimes \Omega^{\dim X ,\bullet}(X) \otimes \Omega^{\bullet,\bullet}(Y)\otimes  \mf h^*[d-2]$, then the action functional is 
\[S(\mc A,\mc B) =  \int_{M\times X\times Y }   \langle \mc B , (d_{M}  +  \ol{\del}_X    + \ol{\del}_Y +  d_{\mf h} )\mc  A  \rangle +  \sum_{n\geq 2} \frac{1}{n!} \int_{M\times X\times Y } \langle \mc  B , \ell_n (\mc A,\cdots,\mc A )\rangle.  \]
In particular, if $X=Y=\pt$ and $(\mf h,[-,-])$ is the Lie algebra of a Lie group $H$, then the action functional simplifies to
\[ S(\mc A,\mc B) = \int_M \langle \mc B,  d_{M} \mc A \rangle + \frac{1}{2} \int_M \langle \mc B, [\mc A,\mc A] \rangle  = \int_M \langle \mc B, F_\mc A\rangle ,  \]
where $\mc A\in \Omega^\bullet(M, \mf h)[1]$, $\mc B \in \Omega^\bullet(M, \mf h^*)[d-2]$, and $F_{\mc A} = d_{M}\mc A + \frac 12 [\mc A,\mc A ]$ is the curvature of $\mc A$. This is a BV description of \emph{($d$-dimensional) (topological) BF theory with gauge group $H$} on $M$.  
\end{ex}

\subsubsection{} A primary goal of this section is to identify twists of pure supersymmetric Yang--Mills theories as a specific class of generalized Chern--Simons theories. For this purpose, we note that a cyclic $L_\infty$-algebra $\mf l$ can be constructed as a tensor product of a cyclic differential graded-commutative algebra $A$  and a cyclic $L_\infty$-algebra $\mf g$  (see \cite[(2.44), (2.45)]{JRSW} for an explicit description). We further specialize our focus to cases where $A$ has trivial differential and  $\mf g$ is the Lie algebra of a Lie group. In particular, $(\mf l=A \otimes \mf g, [-,-]_{\mf l}, \abracket{-,-}_{\mf l} )$ is a cyclic graded Lie algebra.

\begin{defn}\label{defn:enhancedCS}
Let $A$ be a graded-commutative algebra with a cyclic pairing. A \emph{Chern--Simons theory with gauge group $G$ enhanced by $A$ on $M\times X\times Y$} is defined as a generalized Chern--Simons theory with the cyclic $L_\infty$-algebra $\mf l= A \otimes  \mf g  $.
\end{defn}

\begin{rmk}\label{rmk:degree}
If the cyclic pairing of $A$ is of cohomological degree $k$, then the corresponding Chern--Simons theory enhanced by $A$ is defined on $\R^{d_M}\times \C^{d_X}\times \C^{d_Y}$ such that $m=d_M +  d_X + 2 d_Y =k+3$. Because a field theory under this condition has dimension $d = d_M+  2d_ X + 2 d_Y$, we find that $d$ satisfies the inequality $m\leq d\leq 2m $, and for each dimension $d$, there are $\left\lfloor \frac{2m - d}{2} \right\rfloor + 1$ possible theories. Therefore, if we let $T(m)$ denote the number of distinct theories\footnote{These theories share the same type of interactions but differ by their spacetime choices and dependence, as detailed for Chern--Simons theories in Example \ref{ex:CStheories}.} for a fixed $m\geq 1 $, then we have
\[T(m) = \begin{cases}
  \bracket{\frac{m}{2} +1 }^2  & \text{if }m\text{ is even},\\	
  \bracket{\frac{m+1}{2}}\bracket{ \frac{m+3}{2}}  & \text{if }m\text{ is odd}.\\	
 \end{cases}
 \]
For small values of $m$, we obtain $T(1)=2$, $T(2)=4$, $T(3)=6$ (which explains why there are six Chern--Simons theories in Example \ref{ex:CStheories}), $T(4)=9$, and $T(5)=12$.  
\end{rmk}

\begin{rmk}
In fact, if $\mf g= \mf{gl}(N)$, which is the main case of interest for us, an $L_\infty$-algebra structure can be derived from a cyclic $A_\infty$-algebra structure on $A \otimes \mf {gl}(N)$, where $A$ is a cyclic $A_\infty$-algebra (see  \cite{AmorimTuTensor}). We hope to investigate this case further in future work. 
\end{rmk}

\subsection{Examples from 4-dimensional Supersymmetric Yang--Mills Theory}\label{subsection:SUSY}

In this subsection, we express well-known examples of twists of supersymmetric field theories as Chern--Simons theories with gauge group $G$ enhanced by $A$. We focus on 4-dimensional holomorphic examples. For further details, see \cite{CostelloSUSY, ESW}.

\begin{ex}[holomorphic twist of 4-dimensional $\mc N=1$ pure gauge theory]\label{ex:4dN=1} The space of fields of the holomorphic twist of 4-dimensional $\mc N=1$ pure gauge theory with gauge group $G$ on $\C^2$ is
\[ \Omega^{0,\bullet}( \C^2, \mf g )[1] \oplus \Omega^{2,\bullet}(\C^2 ,\mf g^* ),  \]
and the action functional is $ S(\mc A,\mc B ) =  \int _{\C^2 } \abracket{\mc B , \ol{\del} \mc A   + \frac{1}{2} [\mc A,\mc A]    }$, where $\mc A \in \Omega^{0,\bullet}( \C^2, \mf g )[1] $, $\mc B \in  \Omega^{2,\bullet}(\C^2 ,\mf g^* ) $, and $\abracket{-,-}$ denotes the canonical pairing between the Lie algebra $\mf g$ of $G$ and its dual $\mf g^*$.  This is equivalent to (2-dimensional) \emph{holomorphic BF theory} on $\C^2$, corresponding to Example \ref{ex:BF} with $X=\C^2$, $M=Y=\pt$. For more details, see \cite[Section 16]{CostelloSUSY}.

Choosing a holomorphic volume form on $\C^2$, the space of fields can be rewritten as $\Omega^{0,\bullet}( \C^2, \mf g  \oplus\mf g^*[-1] ) [1]$. By choosing a non-degenerate symmetric bilinear form on $\mf g$, we can further rewrite this as 
\[   \Omega^{0,\bullet}( \C^2, \mf g \oplus \mf g[-1] ) [1] \cong    \Omega^{0,\bullet}( \C^2, A \otimes \mf g ) [1] ,  \]
where $A=\C[\eps]$ is a graded-commutative algebra with a variable $\eps$ of cohomological degree $1$ and a trace map $\tr \colon A\to \C$ of degree $-1$, defined by $\tr(\eps)=1$. Thus, the theory is a Chern--Simons theory on $\C^2_{\ol{\del}}$ enhanced by $A=\C[\eps]$.  
\end{ex}

\begin{ex}[holomorphic twist of 4-dimensional $\mc N=2$ pure gauge theory] The space of fields of the holomorphic twist of 4-dimensional $\mc N=2$ pure gauge theory with gauge group $G$ is obtained by replacing $\mf g$ in the $\mc N=1$ theory with $\mf g[\theta ]$, where $\theta$ is a variable of cohomological degree $-1$.  Hence the space of fields is 
\[ \Omega^{0,\bullet}( \C^2, \mf g[\theta] )[1] \oplus \Omega^{2,\bullet}(\C^2  ,(\mf g[\theta])^* ),  \]
and the action functional is $S(\mc A,\mc B ) =  \int _{\C^2 } \abracket{\mc B , \ol{\del} \mc A   + \frac{1}{2} [\mc A,\mc A]    }$, 
where $\mc A \in \Omega^{0,\bullet}( \C^2, \mf g [\theta] )[1] $, $\mc B \in  \Omega^{2,\bullet}(\C^2 ,(\mf g[\theta])^* ) $, and $\abracket{-,-}$ is defined using the pairing between $\mf g$ and $\mf g^*$ and the trace map $\C[\theta]\to \C$ given by $\theta \mapsto 1$. For further details, see \cite{CostelloSUSY}.

Choosing a holomorphic volume form on $\C^2$ and a symmetric bilinear form on $\mf g$, the space of fields can be rewritten as
\[ \Omega^{0,\bullet}( \C^2, \mf g [\theta]  \oplus\mf g^*[-2][\theta] ) [1]   \cong  \Omega^{0,\bullet}( \C^2, A \otimes \mf g ) [1]   \]
where  $A = \C[\theta, \delta ]/(\delta^2)$ is a graded-commutative algebra with an even variable $\delta$ of cohomological degree 2 and a trace map $\tr\colon A\to \C$ of degree $-1$, defined by $\theta \delta\mapsto 1$. Thus, the theory is a Chern--Simons theory on $\C^2_{\ol{\del}}$ enhanced by $A=\C[\eps,\delta]/(\delta^2)$. 
 
\end{ex}

\begin{ex}[holomorphic twist of 4-dimensional $\mc N=4$ gauge theory]	The space of fields of the holomorphic twist of 4-dimensional $\mc N=4$ gauge theory with gauge group $G$ is obtained by replacing $\mf g$ in the $\mc N=1$ theory with $\mf g[\eps_1,\eps_2]$, where $|\eps_i|=1$.  Hence, the space of fields is
\[ \Omega^{0,\bullet}( \C^2, \mf g[\eps_1,\eps_2] )[1] \oplus \Omega^{2,\bullet}(\C^2  ,(\mf g[\eps_1,\eps_2])^* ),  \]
and the action functional is $S(\mc A,\mc B ) =  \int _{\C^2 } \abracket{\mc B , \ol{\del} \mc A   + \frac{1}{2} [\mc A,\mc A]    }$, where $\mc A \in \Omega^{0,\bullet}( \C^2, \mf g [\eps_1,\eps_2] )[1] $, $\mc B \in  \Omega^{2,\bullet}(\C^2 ,(\mf g[\eps_1,\eps_2])^* ) $, and $\abracket{-,-}$ is defined using the pairing between $\mf g$ and $\mf g^*$ and the trace map $\C[\eps_1,\eps_2]\to \C$ given by $\eps_1\eps_2 \mapsto 1$. For more details, see \cite{CostelloSUSY, EY1}.

Unlike the $\mc N=1$ and $\mc N=2$ cases, this theory can be expressed as a Chern--Simons theory without requiring a choice of a holomorphic volume form on $\C^2$. That is, identifying $\eps_i$ as $dz_i$ and choosing a symmetric bilinear form on $\mf g$, the space of fields can be rewritten as
\[  \Omega^{\bullet ,\bullet}( \C^2, \mf g  \oplus\mf g^*[1]  ) [1] \cong  \Omega^{\bullet ,\bullet}( \C^2, A \otimes \mf g  ) [1],  \]
where $A = \C[\theta ]$ is a graded-commutative algebra with a trace map of degree $1$, defined by $\theta \mapsto 1$. Thus, the theory is a Chern--Simons theory on $\C^2_{\mr{Dol}}$ enhanced by $A=\C[\theta]$.
 
\end{ex}

\subsection{Supersymmetric Yang--Mills and Chern--Simons Theories}

In this subsection, we show that the observation from the previous subsection is not a mere coincidence; indeed, we prove that every twist of pure supersymmetric Yang--Mills theory can be described as a generalized Chern--Simons theory.

\begin{prop}\label{prop:twistedSYMisCS}
	Every twist of pure supersymmetric Yang--Mills theory is a Chern--Simons theory enhanced by a cyclic graded-commutative algebra, as summarized in the following tables.
\end{prop}

To elaborate, we use the classification result from \cite{ESW}, which completely characterizes the twists of pure supersymmetric Yang--Mills theory. A twist is specified by a triple $(d,\mc N,\text{twist})$, where $d$ is the spacetime dimension, $\mc N$ is the number of supersymmetries, and the twist is determined by the rank of the supercharge used. This data uniquely characterizes each twist, whose description is summarized in the first four columns of the table below.\footnote{In \cite{ESW}, a section stack is used, but in our context, it coincides with a mapping stack.} The content of the proposition is that each such twist can be reformulated as a Chern--Simons theory on the given spacetime, enhanced by a cyclic graded-commutative algebra $A$, according to the type as specified in the table and explained in the proof. In the table, $\eps,\eps_1,\eps_2,\eps'$ denote odd variables and $\delta$ an even variable.
\begin{table}[htbp]
\centering
\begin{tabular}{c|c|c|c||c|c|c}
\(d\) & \(\mathcal{N}\) & twist & description & spacetime & \(A\) & type \\ \hline
10 & $(1,0)$ & $(1,0)$ & $\umap(\C^5_{\ol{\del}} , B\mf g )$  & $\C^5_{\ol{\del}}$ & $\C$ & \circled{1} \\ \hline
9  & $1$ & $1$ & $\umap(\R_{\mr{dR}}\times  \C^4_{\ol{\del}} , B\mf g )$  & $\R_{\mr{dR}}\times \C^4_{\ol{\del}}$ & $\C$  & \circled{1} \\ \hline
\multirow{3}{*}{8}  & \multirow{3}{*}{$1$} & $(1,0)$ pure &  $T^*[-1]\umap(\C^4_{\ol{\del}} , B\mf g )$  & $\C^4_{\ol{\del}}$ & $\C[\eps]$  & \circled{2} \\ 
    \cline{3-7}
   &     & $(1,1)$      & $\umap(\R_{\mr{dR}}^2\times  \C^3_{\ol{\del}} , B\mf g )$   & $\R^2_{\mr{dR}}\times \C^3_{\ol{\del}}$ & $\C$  & \circled{1}  \\ 
    \cline{3-7}
   &     & $(1,0)$ impure & $\umap(\C^4_{\ol{\del}} , B\mf g )_{\mr{dR}}$  & $\C^4_{\ol{\del}}$ & $ (\C[\eps],\,\frac{\del}{\del\eps} )$  & \circled{9} \\ \hline
\multirow{3}{*}{7}  & \multirow{3}{*}{$1$} & $1$ pure & $T^*[-1]\umap(\R_{\mr{dR}}\times  \C^3_{\ol{\del}} , B\mf g )$  & $\R_{\mr{dR}}\times \C^3_{\ol{\del}}$ & $\C[\eps]$  & \circled{2} \\ 
    \cline{3-7}
   &     & $2$ & $\umap(\R_{\mr{dR}}^3\times  \C^2_{\ol{\del}} , B\mf g )$   & $\R^3_{\mr{dR}}\times \C^2_{\ol{\del}}$ & $\C$  & \circled{1} \\ 
    \cline{3-7}
   &     & $1$ impure &$\umap(\R_{\mr{dR}}\times  \C^3_{\ol{\del}} , B\mf g )_{\mr{dR}}$   & $\R_{\mr{dR}}\times \C^3_{\ol{\del}}$ & $ (\C[\eps],\,\frac{\del}{\del\eps} )$ & \circled{9} \\ \hline
\multirow{4}{*}{6}  & \multirow{4}{*}{$(1,1)$} & $(1,0)$ &  $T^*[-1]\umap(  \C^3_{\ol{\del}} , \mf g/\mf g )$   & $\C^3_{\ol{\del}}$ & $\C[\eps_1,\eps_2]$  & \circled{7}\\ 
     \cline{3-7}
  &         & $(1,1)$ special & $T^*[-1]\umap(\R_{\mr{dR}}^2\times  \C^2_{\ol{\del}} , B\mf g )$   & $\R^2_{\mr{dR}}\times \C^2_{\ol{\del}}$ & $\C[\eps]$ & \circled{2} \\ 
     \cline{3-7}
  &         & $(2,2)$ & $\umap(\R_{\mr{dR}}^4\times  \C_{\ol{\del}} , B\mf g )$  & $\R^4_{\mr{dR}}\times \C_{\ol{\del}}$ & $\C$  & \circled{1}\\ 
    \cline{3-7}
   &         & $(1,1)$ generic & $\umap(\R_{\mr{dR}}^2\times  \C^2_{\ol{\del}} , B\mf g )_{\mr{dR}}$   & $\R^2_{\mr{dR}}\times \C^2_{\ol{\del}}$ & $ (\C[\eps],\,\frac{\del}{\del\eps} )$ & \circled{9} \\ \hline
\multirow{4}{*}{5}  & \multirow{4}{*}{$2$} & $1$ & $T^*[-1]\umap(\R_{\mr{dR}}\times  \C^2_{\ol{\del}} , \mf g/\mf g )$  & $\R_{\mr{dR}}\times \C^2_{\ol{\del}}$ & $\C[\eps_1,\eps_2]$  & \circled{7} \\ 
    \cline{3-7}
   &     & $2$ special & $T^*[-1]\umap(\R_{\mr{dR}}^3\times  \C_{\ol{\del}} , B\mf g )$ & $\R^3_{\mr{dR}}\times \C_{\ol{\del}}$ & $\C[\eps]$  & \circled{2} \\ 
    \cline{3-7}
   &     & $4$ & $\umap(\R_{\mr{dR}}^5 , B\mf g )$   & $\R^5_{\mr{dR}}$ & $\C$  & \circled{1} \\ 
    \cline{3-7}
   &     & $2$ generic & $\umap(\R_{\mr{dR}}^3\times  \C_{\ol{\del}} , B\mf g )_{\mr{dR}}$   & $\R^3_{\mr{dR}}\times \C_{\ol{\del}}$ & $ (\C[\eps],\,\frac{\del}{\del\eps} )$  & \circled{9} \\ \hline
\multirow{6}{*}{4}  & \multirow{6}{*}{$4$} & $(1,0)$ & $T^*[-1]\umap( \C^2_{\mr{Dol}} , B\mf g )$   & $\C^2_{\mr{Dol}}$ & $\C[\eps]$  & \circled{2} \\ 
    \cline{3-7}
   &     & $(1,1)$ & $T^*[-1]\umap(\R_{\mr{dR}}^2\times  \C_{\mr{Dol}} , B\mf g )$   & $\R^2_{\mr{dR}}\times \C_{\mr{Dol}}$ & $\C[\eps]$  & \circled{2} \\ 
    \cline{3-7}
   &     & $(2,2)$ special& $T^*[-1]\umap(\R_{\mr{dR}}^4 , B\mf g )$   & $\R^4_{\mr{dR}}$ & $\C[\eps]$  & \circled{2} \\ 
    \cline{3-7}
   &     & $(2,0)$ &   $\umap(\C^2_{\mr{Dol}} , B\mf g )_{\mr{dR}}$ &   $\C^2_{\mr{Dol}}$  & $ (\C[\eps],\,\frac{\del}{\del\eps} )$  & \circled{9}\\ 
    \cline{3-7}
   &     & $(2,1)$ &$\umap(\R_{\mr{dR}}^2\times  \C_{\mr{Dol}} , B\mf g )_{\mr{dR}} $ & $\R^2_{\mr{dR}}\times \C_{\mr{Dol}}$ & $ (\C[\eps],\,\frac{\del}{\del\eps} )$  & \circled{9}\\ 
    \cline{3-7}
   &     & $(2,2)$ generic & $\umap(\R_{\mr{dR}}^4, B\mf g )_{\mr{dR}}$   & $\R^4_{\mr{dR}}$ & $ (\C[\eps],\,\frac{\del}{\del\eps} )$  & \circled{9} \\ \hline
\multirow{3}{*}{3}  & \multirow{3}{*}{$8$} & $1$ & $T^*[-1]\umap(\R_{\mr{dR}}\times  \C_{\mr{Dol}} , \mf g/\mf g )$   & $\R_{\mr{dR}}\times \C_{\mr{Dol} }$ & $\C[\eps,\eps']$ & \circled{7} \\ 
    \cline{3-7}
   &     & $2$ (B) & $T^*[-1]\umap(\R_{\mr{dR}}^3 , \mf g/\mf g )$ & $\R^3_{\mr{dR}}$ & $\C[\eps,\eps']$ & \circled{7} \\ 
    \cline{3-7}
   &     & $2$ (A) & $ \umap(\R_{\mr{dR}}^3 , \mf g/\mf g )_{\mr{dR}}$ & $\R^3_{\mr{dR}}$ & $ (\C[\eps,\eps'],\,\frac{\del}{\del\eps} )$  & \circled{9}\\ \hline
\end{tabular}
\caption{Twists of maximally supersymmetric Yang--Mills as Chern--Simons enhanced by $A$}
\label{table1_CS}
\end{table}

\begin{table}
\centering
\begin{tabular}{c|c|c|c||c|c|c}
\(d\) & \(\mathcal{N}\) & twist & description  & spacetime & \(A\) & type \\ \hline
\multirow{1}{*}{6}  
   & \multirow{1}{*}{\((1,0)\)} & \((1,0)\) pure & $\umap( \C^3 _{\ol{\del}} , 0\sslash \mf g )$ & \(\C^3_{\ol{\del}}\) & \(\C[\delta]/(\delta^2)\)  & \circled{4} \\ \hline
\multirow{1}{*}{5}  
   & \multirow{1}{*}{\(1\)}   & \(1\) &  $\umap( \R_{\mr{dR}}\times  \C^2 _{\ol{\del}} , 0\sslash \mf g )$  & \(\R_{\mr{dR}} \times \C^2_{\ol{\del}}\) & \(\C[\delta]/(\delta^2)\) & \circled{4}  \\ \hline
\multirow{3}{*}{4} 
   & \multirow{3}{*}{\(2\)} & \((1,0)\) &  $T^*[-1]\umap( \C^2 _{\ol{\del}} , 0\sslash \mf g )$  & \(\C^2_{\ol{\del}}\) & \(\C[\eps,\delta]/(\delta^2)\) & \circled{5}  \\ 
   \cline{3-7}
   &                         & \((1,1)\) &   $\umap(\R^2_{\mr{dR}}\times  \C _{\ol{\del}} , 0\sslash \mf g )$  & \(\R^2_{\mr{dR}} \times \C_{\ol{\del}}\) & \(\C[\delta]/(\delta^2)\) & \circled{4}  \\
    \cline{3-7}
   &                         & \((2,0)\) &   $\umap( \C^2 _{\ol{\del}} , 0\sslash \mf g )_{\mr{dR}}$  & \(\C^2_{\ol{\del}}\) & \( (\C[\eps,\delta]/(\delta^2),\,\frac{\del}{\del \eps} )\)& \circled{9}  \\ \hline
\multirow{3}{*}{3} 
   & \multirow{3}{*}{\(4\)}  
   & \(1\) &   $T^*[-1]\umap( \R_{\mr{dR}}\times \C _{\ol{\del}} , 0\sslash \mf g )$  & \(\R_{\mr{dR}} \times \C_{\ol{\del}}\) & \(\C[\eps,\delta]/(\delta^2)\) & \circled{5}  \\
    \cline{3-7}
   &                         & \(2\) (B) &   $\umap( \R^3 _{\mr{dR}} , 0\sslash \mf g )$  & \(\R^3_{\mr{dR}}\) & \(\C[\delta]/(\delta^2)\) & \circled{4}  \\ 
   \cline{3-7}
   &                         & \(2\) (A) &  $\umap( \R_{\mr{dR}}\times \C _{\ol{\del}} , 0\sslash \mf g )_{\mr{dR}}$ & \(\R_{\mr{dR}} \times \C_{\ol{\del}}\) & \( (\C[\eps,\delta]/(\delta^2),\,\frac{\del}{\del \eps} )\)& \circled{9}  \\ \hline
\end{tabular}
\caption{Twists of super Yang--Mills with 8 supercharges as Chern--Simons enhanced by $A$}
\label{table2_CS}
\end{table}

\begin{table}
\centering
\begin{tabular}{c|c|c|c||c|c|c}
\(d\) & \(\mathcal{N}\) & twist & description & spacetime & \(A\)   & type \\ \hline
\multirow{1}{*}{4}  
   & \multirow{1}{*}{\(1\)} & \((1,0)\) & $T^*[-1]\umap( \C^2_{\ol{\del}} , B\mf g )$ & \(\C^2_{\ol{\del}}\) & \(\C[\eps]\) & \circled{2}  \\
    \hline
 {3}  &   {\(2\)} & \(1\) & $T^*[-1]\umap(\R_{\mr{dR}} \times  \C_{\ol{\del}} , B\mf g )$ & \(\R_{\mr{dR}} \times \C_{\ol{\del}}\) & \(\C[\eps]\) & \circled{2} \\ 
   \hline 
\end{tabular}
\caption{Twists of super Yang--Mills with 4 supercharges as Chern--Simons enhanced by $A$}
\label{table3_CS}
\end{table}

\begin{table}
\centering
\begin{tabular}{c|c|c||c|c|c}
\(\mathcal{N}\) & twist & description  & spacetime & \(A\)& type  \\ \hline
\multirow{3}{*}{\((4,4)\)} 
  & \((1,0)\) & $T^*[-1]\umap( \C_{\ol{\del}} , T[1]( 0 \sslash \mf g ) )$ & \(\C_{\ol{\del}}\) & \(\C[\eps,\eps',\delta]/(\delta^2)\)  & \circled{8}\\
   \cline{2-6}
  & \((1,1)\) (B) & $T^*[-1]\umap( \R^2_{\mr{dR}},  0\sslash \mf g )$ & \(\R^2_{\mr{dR}}\) & \(\C[\eps,\delta]/(\delta^2)\) & \circled{5} \\ \cline{2-6}
  & \((1,1)\) (A) & $\umap( \R^2_{\mr{dR}},  (0\sslash \mf g)_{\mr{dR}} )$ & \(\R^2_{\mr{dR}}\) & \( (\C[\eps,\delta]/(\delta^2),\,\frac{\del}{\del \eps} )\)  & \circled{9} \\ \hline
\multirow{3}{*}{\((2,2)\)}
  & \((1,0)\) & $T^*[-1]\umap( \C_{\ol{\del}} , T[1] B \mf g  )$  & \(\C_{\ol{\del}}\) & \(\C[\eps,\eps']\) & \circled{8} \\
   \cline{2-6}
  & \((1,1)\) (B) & $T^*[-1]\umap( \R^2_{\mr{dR}} , B \mf g  )$ & \(\R^2_{\mr{dR}}\) & \(\C[\eps']\)  & \circled{2} \\ 
  \cline{2-6}
  & \((1,1)\) (A) & $T^*[-1]\umap( \C_{\ol{\del}} ,  (B \mf g)_{\mr{dR}}  )$ & \(\C_{\ol{\del}}\) & \( (\C[\eps,\eps'],\,\frac{\del}{\del \eps} )\) & \circled{9} \\ \hline
  $(4,0)$ & $(1,0)$ &   $T^*[-1]\umap( \C_{\ol{\del}} ,  0\sslash \mf g  )$ & \(\C_{\ol{\del}}\) & \(\C[\delta_1,\delta_2]/(\delta^2_1,\delta_2^2)\)  & \circled{6} \\ \hline
  $(2,0)$ & $(1,0)$ &   $T^*[-1]\umap( \C_{\ol{\del}} ,  B \mf g  )$ & \(\C_{\ol{\del}}\) & \(\C[\delta]/(\delta^2)\)  & \circled{3} \\ \hline
\end{tabular}
\caption{Twists of super Yang--Mills in two dimensions as Chern--Simons enhanced by $A$}
\label{table4_CS}
\end{table}

\FloatBarrier

\begin{proof}
Generalized Chern--Simons theories are defined by maps into the classifying stack of a cyclic $L_\infty$-algebra. Hence, we must recast both cotangent stacks of mapping stacks and mapping stacks into symplectic reductions in this same language. To that end, we make the following observations from derived algebraic geometry:
 \begin{itemize}
	\item [(a)] If $\mc Y$ is $m$-oriented and $\mf l$ is an $L_\infty$-algebra, then there exists an equivalence \[ T^*[k] \umap(\mc Y,B \mf l )\simeq \umap( \mc Y, T^*[k+m]B \mf l ) \] as shown in  \cite[Corollary 1.12]{ESW}.
	\item [(b)] If $B\mf g$ is 2-shifted symplectic (e.g., if $\mf g$ is the Lie algebra of a reductive group $G$), then we have \[T^*[1]B\mf g\cong T[-1]B\mf g  \cong  B(\C[\eps]  \otimes  \mf g )\] where $\eps$ is a variable of cohomological degree 1. More generally, if $B\mf l$ is shifted symplectic of even degree (e.g., if $\mf l =A \otimes \mf g$ where $A$ is a graded-commutative algebra with a cyclic trace map of even degree), then we have \[T^*[2k+1]B\mf l \cong T[2l+1]B\mf l \cong B( \C[\eps]\otimes \mf l)\] in a $\Z/2\Z$-graded sense for any integers $k,l$.
	\item [(c)]  Since there is an isomorphism between the derived symplectic reduction $(T^*\mc X )\sslash G: = \R \mu^{-1}(0)/G$ and the cotangent stack $ T^*(\mc X /G) $ for a derived stack $\mc X$ and an affine group scheme $G$, it follows that, if $B\mf g$ is 2-shifted symplectic, then we have \[0\sslash \mf g \cong T^*B\mf g \cong T[-2]B\mf g  \cong B \left (\C[\delta]/(\delta^2) \otimes \mf g \right )  \]
	 where $\delta$ is a variable of cohomological degree 2.	 More generally, if $B\mf l$ is shifted symplectic of even degree, then we have \[T^*[2k]B\mf l \cong T[2l]B\mf l \cong B\left ( \C[\delta] / (\delta^2) \otimes \mf l\right )\] in a $\Z/2\Z$-graded sense for any integers $k,l$. On the other hand, if $B\mf l$ is shifted symplectic of odd degree, then we have \[T^*[2k]B\mf l \cong T[2l-1]B\mf l \cong B\left ( \C[\eps] \otimes \mf l\right )\] in a $\Z/2\Z$-graded sense for any integers $k,l$. 

	\item [(d)] Since there is an isomorphism $G/G \cong \umap(S^1_{\text{B}},BG)$ for an algebraic group $G$, we also have \[\mf g/\mf g \cong \umap(S^1_{\text{B}},B\mf g) \cong B( \C[\eps]\otimes \mf g  )  \]
	where $\eps$ is a variable of cohomological degree 1.
	\item [(e)] The complex $\mf g_{\mr{dR}}$ is defined as
	\[\xymatrix@R-20pt{
\ul{-1}& \ul{0}\\
\mf g \ar[r]^-{\id} & \mf g .
	}\] 
	Hence, we have an isomorphism $\mf g_{\mr{dR}}\cong A \otimes \mf g$, where $A=\left (\C[\theta ] , \frac{\del}{\del \theta } \right )$ with a variable $\theta$ of cohomological degree $-1$. In a $\Z/2\Z$-graded setting, we write $\eps$ instead of $\theta$.	
		\end{itemize}
Every description in the tables falls into one of the following nine types, each obtained by combining the observations above:
\begin{enumerate}
	\item [\circled{1}] \fbox{If $\umap(M_{\mr{dR}}\times X_{\ol{\del}}\times Y_{\mr{Dol}},B\mf g )$:} These are by definition Chern--Simons theories on $M_{\mr{dR}}\times X_{\ol{\del}}\times Y_{\mr{Dol}}$ with gauge group $G$ enhanced by $A=\C$.
	\item [\circled{2}] \fbox{If $T^*[-1]\umap(M_{\mr{dR}}\times X_{\ol{\del}}\times Y_{\mr{Dol}},B\mf g )$, where $M_{\mr{dR}}\times X_{\ol{\del}}\times Y_{\mr{Dol}}$ is $2k$-oriented:} We have
	\begin{align*}
&T^*[-1]\umap(M_{\mr{dR}}\times X_{\ol{\del}}\times Y_{\mr{Dol}},B\mf g ) \stackrel{\text{(a)}}{\cong}  \umap(M_{\mr{dR}}\times X_{\ol{\del}}\times Y_{\mr{Dol}}, T^*[2k-1] B\mf g )\\
 & \stackrel{\text{(b)}}{\cong} \umap(M_{\mr{dR}}\times X_{\ol{\del}}\times Y_{\mr{Dol}},  B(\C[\eps] \otimes \mf g )) 
 	\end{align*}
Hence, these are Chern--Simons theories on $M_{\mr{dR}}\times X_{\ol{\del}}\times Y_{\mr{Dol}}$ enhanced by $A=\C[\eps]$.
	\item [\circled{3}] \fbox{If $T^*[-1]\umap(M_{\mr{dR}}\times X_{\ol{\del}}\times Y_{\mr{Dol}},B\mf g )$, where $M_{\mr{dR}}\times X_{\ol{\del}}\times Y_{\mr{Dol}}$ is $(2k+1)$-oriented:} We have
	\begin{align*}
&T^*[-1]\umap(M_{\mr{dR}}\times X_{\ol{\del}}\times Y_{\mr{Dol}},B\mf g ) \stackrel{\text{(a)}}{\cong}  \umap(M_{\mr{dR}}\times X_{\ol{\del}}\times Y_{\mr{Dol}}, T^*[2k] B\mf g )\\
 & \stackrel{\text{(c)}}{\cong} \umap \left (M_{\mr{dR}}\times X_{\ol{\del}}\times Y_{\mr{Dol}},  B(\C[\delta]/(\delta^2) \otimes \mf g ) \right ) 
 	\end{align*}
Hence, these are Chern--Simons theories on $M_{\mr{dR}}\times X_{\ol{\del}}\times Y_{\mr{Dol}}$ enhanced by $A=\C[\delta]/(\delta^2)$.
	\item [\circled{4}] \fbox{If $\umap(M_{\mr{dR}}\times X_{\ol{\del}}\times Y_{\mr{Dol}}, 0\sslash \mf g )$:} We have 
	\[  \umap(M_{\mr{dR}}\times X_{\ol{\del}}\times Y_{\mr{Dol}}, 0\sslash \mf g )\stackrel{\text{(c)}}{\cong}  \umap\left (M_{\mr{dR}}\times X_{\ol{\del}}\times Y_{\mr{Dol}},B    (\C[\delta]/(\delta^2) \otimes \mf g   )  \right ) \]
Hence, these are Chern--Simons theories on $M_{\mr{dR}}\times X_{\ol{\del}}\times Y_{\mr{Dol}}$ enhanced by $A=\C[\delta]/(\delta^2) $.
	\item [\circled{5}] \fbox{If $T^*[-1]\umap(M_{\mr{dR}}\times X_{\ol{\del}}\times Y_{\mr{Dol}}, 0\sslash \mf g )$, where $M_{\mr{dR}}\times X_{\ol{\del}}\times Y_{\mr{Dol}}$ is $2k$-oriented:} We have 
	\begin{align*}
&T^*[-1]\umap(M_{\mr{dR}}\times X_{\ol{\del}}\times Y_{\mr{Dol}},0\sslash \mf g ) \stackrel{\text{(a),(c) } }{\cong}  \umap\left (M_{\mr{dR}}\times X_{\ol{\del}}\times Y_{\mr{Dol}}, T^*[2k-1]B (\C[\delta]/(\delta^2) \otimes \mf g  )  \right )\\
 & \stackrel{\text{(b)}}{\cong} \umap\left (M_{\mr{dR}}\times X_{\ol{\del}}\times Y_{\mr{Dol}},  B(\C[\eps,\delta]/(\delta^2) \otimes \mf g ) \right  ) 
 	\end{align*}
Hence, these are Chern--Simons theories on  $M_{\mr{dR}}\times X_{\ol{\del}}\times Y_{\mr{Dol}}$ enhanced by $A=\C[\eps,\delta]/(\delta^2)$.
	\item [\circled{6}] \fbox{If $T^*[-1]\umap(M_{\mr{dR}}\times X_{\ol{\del}}\times Y_{\mr{Dol}}, 0\sslash \mf g )$, where $M_{\mr{dR}}\times X_{\ol{\del}}\times Y_{\mr{Dol}}$ is $(2k+1)$-oriented:} We have 
	\begin{align*}
&T^*[-1]\umap(M_{\mr{dR}}\times X_{\ol{\del}}\times Y_{\mr{Dol}},0\sslash \mf g ) \stackrel{\text{(a),(c) } }{\cong}  \umap\left (M_{\mr{dR}}\times X_{\ol{\del}}\times Y_{\mr{Dol}}, T^*[2k]B (\C[\delta]/(\delta^2) \otimes \mf g  )  \right )\\
 & \stackrel{\text{(c)}}{\cong} \umap\left (M_{\mr{dR}}\times X_{\ol{\del}}\times Y_{\mr{Dol}},  B(\C[\delta_1,\delta_2]/(\delta_1^2,\delta_2^2) \otimes \mf g ) \right  ) 
 	\end{align*}
Hence, these are Chern--Simons theories on  $M_{\mr{dR}}\times X_{\ol{\del}}\times Y_{\mr{Dol}}$ enhanced by $A=\C[\delta_1,\delta_2]/( \delta_1^2,\delta_2^2)$.
	\item [\circled{7}] \fbox{If $T^*[-1]\umap(M_{\mr{dR}}\times X_{\ol{\del}}\times Y_{\mr{Dol}}, \mf g/  \mf g )$, where $M_{\mr{dR}}\times X_{\ol{\del}}\times Y_{\mr{Dol}}$ is $(2k+1)$-oriented:} We have 
	\begin{align*}
&	T^*[-1]\umap(M_{\mr{dR}}\times X_{\ol{\del}}\times Y_{\mr{Dol}}, \mf g/  \mf g ) \stackrel{\text{(a),(d) } }{\cong} 	\umap(M_{\mr{dR}}\times X_{\ol{\del}}\times Y_{\mr{Dol}}, T^*[2k]  B( \C[\eps]\otimes \mf g ) )\\
 & \stackrel{\text{(c)}}{\cong} \umap\left (M_{\mr{dR}}\times X_{\ol{\del}}\times Y_{\mr{Dol}},  B(\C[\eps,\eps']  \otimes \mf g ) \right  ) 
	\end{align*}
Hence, these are Chern--Simons theories on  $M_{\mr{dR}}\times X_{\ol{\del}}\times Y_{\mr{Dol}}$ enhanced by $A=\C[\eps,\eps']$.
	\item [\circled{8}] \fbox{If $T^*[-1]\umap( \C_{\ol{\del}} , T[1]( 0 \sslash \mf g ) )$ or $T^*[-1]\umap( \C_{\ol{\del}} , T[1] B\mf g )$:} We have
	\begin{align*}
&	T^*[-1]\umap( \C_{\ol{\del}} , T[1]( A \otimes \mf g ) )	\stackrel{\text{(a)}}{\cong } \umap\left ( \C_{\ol{\del}} ,   T^*B( A[\eps]  \otimes \mf g ) \right)\\
& \stackrel{!}{\cong}  \umap\left ( \C_{\ol{\del}} ,   T[2k+1]B( A[\eps]  \otimes \mf g ) \right) \cong  \umap\left ( \C_{\ol{\del}} ,   B( A[\eps,\eps' ]  \otimes \mf g ) \right)
	\end{align*}
where the isomorphism marked by $\stackrel{!}{\cong}$ holds because $B(A[\eps]\otimes \mf g)$ is shifted symplectic of odd degree. Hence, there are Chern--Simons theories on $\C_{\ol{\del}}$ enhanced by $ A[\eps,\eps' ] $, where $A=\C[\delta]/(\delta^2)$ or $A=\C$.
	\item [\circled{9}] \fbox{If we have the de Rham stack of a mapping stack, or when the target is a de Rham stack:} \\ The cyclic $L_\infty$-algebra is quasi-isomorphic to the zero complex. Consequently, these Chern--Simons theories are perturbatively trivial. Similarly, if $A$ is cohomologically trivial, then the corresponding Chern--Simons theories are perturbatively trivial. Therefore, the proposition follows immediately in these cases.
\end{enumerate}
	\end{proof}

\begin{rmk}
In fact, the case of $\circled{9}$ has a bit more content in that these are written in a way that respects the entire Hodge family. For more details on the appearance of Hodge stack in the BV formalism, see \cite[Example 2.47]{EY1} or \cite[Subsection 1.6.3]{ESW}.
\end{rmk}

\section{Generalized Chern--Simons Theories in Topological Strings}

\subsection{Topological Strings and Open-String Field Theories} 

In this subsection, we introduce a particular class of topological string theories and discuss open-string field theories in that context. This setting is motivated by the foundational work of Costello and Li \cite{CL16}, which conjectures that these topological string theories arise as twists of Type II superstring theories. For further context, see \cite{CL16}, our joint work with Raghavendran \cite{RY1}, and references therein. For an introduction to the use of Ext groups in the context of topological strings aimed at physicists, see \cite{KatzSharpe, Sharpe03, Aspinwall05}. Other discussion of open-string field theories in topological string theory in the classical BV formalism, particularly aimed at physicists, can be found in \cite{EagerSaberi, BGKWWY}.

\subsubsection{} 

In this paper, we consider topological strings described by a 2-dimensional fully extended topological quantum field theory (TQFT), which is determined by a Calabi--Yau (CY) category $\mc C$ of dimension 5. This is based on the fact that a CY category gives rise to a 2-dimensional TQFT, a result first established by Costello and later given a more conceptual treatment by Hopkins and Lurie \cite{CostelloCY, LurieTQFT}.

Let $M$ be a symplectic manifold of dimension $2m$, which we assume throughout to be  $\R^{2m}$, and let $X$ be a smooth (not necessarily proper) CY manifold of dimension $n$, such that $2m+2n=10$. We then consider a CY category of dimension 5 of the form
\[\text{Fuk}(M)\otimes \coh(X),    \]
where $\text{Fuk}(M)$ is the differential graded (DG) ``Fukaya category'' and $\coh(X)$ is the DG category of coherent sheaves on $X$. The category $\text{Fuk}(\R^{2m})$ is generated by a Lagrangian submanifold $\R^m\subset \R^{2m}$ and it has morphism complexes $\uhom_{ \text{Fuk}(\R^{2m}) }( \R^m,\R^m ): =(  \Omega^\bullet(\R^m),d )$. The tensor product $\otimes $ in the 2-category of DG categories is defined as follows: for DG categories $\mc C$ and $\mc D$, the expression $\mc E\otimes \mc F$ is an object of $\mc C \otimes \mc D$, where $\mc E \in\mc C $ and $\mc F\in\mc D $, and the morphism complexess are given by 
\begin{equation}\label{eqn:tensor}
\ul \Hom_{ \mc  C\otimes \mc D } ( \mc E\otimes \mc F,\mc E'\otimes \mc F' ) \cong  \ul \Hom_{\mc C}(\mc E,\mc E')  \otimes \ul \Hom_{\mc D}(\mc F,\mc F').	
\end{equation}

\begin{rmk}\label{rmk:CLconjecture}
Costello and Li \cite{CL16} developed the idea of twisting in string theory and supergravity. Moreover, they conjectured that twisted versions of Type II string theories correspond to topological string theories of the type described above. Specifically, they argued that twists of Type IIB string theory are $\mr{Fuk}( M) \otimes \coh(X)$ when $\dim X$ is odd; henceforth we denote such topological strings by $\IIB[M_A \times X_B]$. Similarly, twists of Type IIA string theory are $\mr{Fuk}( M) \otimes \coh(X)$ when $\dim X$ is even and denoted by $\IIA[M_A \times X_B]$. In summary, we consider twists of Type II string theories corresponding to 2-dimensional extended TQFTs determined by $\mathrm{Fuk}(M) \otimes \coh(X)$, with the dimensions of $M$ and $X$ varying accordingly. We will use II[$M_A \times X_B$] to refer to both IIA and IIB twists.
\end{rmk}

\subsubsection{}
Open-string field theory refers to a field theory that describes the full dynamics of open strings, including both massless and massive modes, which correspond to different excitations of the strings. In the low-energy limit, when one restricts to the massless sector, this reduces to the effective world-volume gauge theory on D-branes, which is typically a supersymmetric Yang--Mills theory.

In topological string theory described by a CY category $\mc C$ of dimension 5, a \emph{D-brane} refers to an object $\mc F \in \mc C$, where the dimension of its support determines the world-volume dimension of the corresponding D$k$-brane. If $\mc F$ has a support on a $(k+1)$-dimensional space, then it is referred to as a D$k$-brane. For instance, $\mc F= \R^2 \otimes \mc O_{\C^2}\in \mr{Fuk}(\R^4)\otimes \coh(\C^3)$ is a D5-brane.

\begin{rmk}
	It is easy to see that $\IIB[M_A \times X_B]$ only admits supersymmetric D$(2p+1)$-branes, whereas $\IIA[M_A \times X_B]$ only admits supersymmetric D$(2p)$-branes. This is precisely as expected from the physical string theory. Hence, this serves as an initial consistency check for the conjecture of Costello and Li.
\end{rmk}

 In the context of topological strings --- which captures only the massless sector --- open-string field theory is the world-volume gauge theory on D-branes, describing interactions between open strings ending on these D-branes. Because this field theory encodes the states of open strings with boundary conditions given by $\mc F$ at each end, we consider the self-Hom complex $\uhom_{\mc C }( \mc F,\mc F )$, which describes these states as morphisms between $\mc F$ and itself. This complex carries an algebra structure that models the interaction of open strings, represented geometrically by the pair-of-pants diagram, where two incoming strings merge into one. Consequently, the self-Ext algebra $\ext_{\mc C }(\mc F,\mc F):=H^\bullet\left( \uhom_{\mc C }( \mc F,\mc F )\right )$ inherits an $A_\infty$-algebra structure. If $\mc F\in \mc C$ is chosen so that $\ext_{\mc C }(\mc F,\mc F)$ is finite-dimensional, then the CY structure of $\mc C$ induces a trace map, which endows $\ext_{\mc C }(\mc F,\mc F)$ with a \emph{cyclic $A_\infty$-algebra} structure.
 
To build the corresponding field theory from this data, the Ext algebra $\ext_{\mc C }(\mc F,\mc F)$ must be promoted to a sheaf of sections of a graded vector bundle over spacetime, which then we denote by  $\ul \ext_{\mc C }(\mc F,\mc F)$. Since we understand the tensor product of DG categories by \eqref{eqn:tensor}, it is enough to understand each case of $\mr{Fuk}(\R^{2m})$ and $\coh(X)$ separately.

The case of $\mr{Fuk}(\R^{2m})$ is relatively straightforward, as the object $\R^m\in  \mr{Fuk}(\R^{2m})$ generates the category. Concretely, we set
\begin{equation}\label{eqn:Fuk}
	\ul\ext_{ \text{Fuk}(\R^{2m}) }( \R^m,\R^m) := \uhom_{ \text{Fuk}(\R^{2m}) }( \R^m,\R^m ) =(  \Omega^\bullet(\R^m),d ).
\end{equation}
This has a cyclic graded-commutative algebra structure, or a trivial cyclic local $L_\infty$-algebra structure on $\R^m$, and hence defines a classical field theory on $\R^m$ in a $\Z/2\Z$-graded sense when $m$ is odd.

\begin{rmk}
Due to non-perturbative instanton effects, one cannot define such an open-string field theory as a classical field theory for a general symplectic manifold $M$. This explains our choice of exclusively working with $\mr{Fuk}(\R^{2m})$ in the definition of topological strings. For further discussion on this topic, we refer to \cite[Remark 3.6]{RY1}.
\end{rmk}

On the other hand, if $\mc F  \in \coh(X) $ is a sheaf of sections of a vector bundle on a complex submanifold $Z\subset X$, then the Ext algebra $\ext_{\coh(X) }(\mc F,\mc F)$ can be promoted to a sheaf of DG algebras of sections of a vector bundle on $Z$.  To be more concrete, consider $\mc F= \mathcal{O}_Z\in \coh (X)$ supported on $Z \subset X$. Then, by \cite[(29)]{KatzSharpe}, we have 
\begin{equation}\label{eqn:coh}
\ul \ext_{\coh(X)} (\mc O_Z,\mc O_Z) := \left( \Omega^{0,\bullet}(Z, \wedge^\bullet \mc N_{Z/X} ) , \ol{\del} \right)  
\end{equation}
where $\mc N_{Z/X}$ is the normal bundle of $Z$ in $X$ and $ \ol{\del} $ is the Dolbeault operator associated with the holomorphic structure of $ \wedge^\bullet \mc N_{Z/X}$. This has a graded-commutative algebra structure, or a trivial local $L_\infty$-algebra structure, on $Z$. Furthermore, it is equipped with a cyclic pairing induced from the CY structure of $\coh(X)$, yielding a cyclic local $L_\infty$-algebra on $Z$ and hence a classical field theory on $Z$ in a $\Z/2\Z$-graded sense if $\dim_\C X $ is odd. 

\begin{defn}
	In topological string theory II$[\R^{2m}_A\times X_B]$, the \emph{open-string field theory} of a D-brane on $\R^m\times Z $ is a classical field theory on $\R^{m}\times Z$, described by $\ul\ext_{\mr{Fuk}(\R^{2m})\otimes \coh(X) }(\R^m\otimes \mc O_Z,\R^m\otimes \mc O_Z )$.
\end{defn}

\begin{rmk}
Since $m+\dim_\C X =5$, the resulting cyclic pairing has odd cohomological degree, so it always defines a classical field theory in a $\Z/2\Z$-graded sense.
	
\end{rmk}

\begin{rmk}
Consider a stack of $N$ D-branes, represented by $\mc F^{\oplus N} \in \mc C$. Its Ext algebra is given by
\[\ul\ext_{\mc C }(\mc F^{\oplus N},\mc F^{\oplus N}) \cong  \ul\ext_{\mc C }(\mc F,\mc F) \otimes \mf{gl}(N), \]
where its cyclic graded Lie algebra structure is precisely the one of the tensor product of a cyclic differential graded-commutative algebra $\ul\ext_{\mc C }(\mc F,\mc F)$ and a Lie algebra $\mf{gl}(N)$, precisely as for a Chern--Simons theory enhanced by a cyclic graded-commutative algebra. Thus, for the purpose of showing that a Chern--Simons theory with gauge group $\GL(N)$ is an open-string field theory, it is enough to assume $N=1$, which explains our choice of exclusively working with $N=1$ in the definition.
\end{rmk}

\begin{rmk}
The Calabi--Yau structure on $X$ endows $\ul \ext_{\coh(X)} (\mc O_Z,\mc O_Z)$ with a Calabi--Yau structure. In particular, if $Z$ is Calabi--Yau itself, then a choice of a volume form on $Z$ identifies the corresponding open-string field theory with Chern--Simons theory; we will suppress this choice in what follows.
\end{rmk}

\begin{ex}[open-string field theory for a D5-brane on $ \R^2 \times  \mathbb{C}^2 \subset \R^4 \times \mathbb{C}^3$] \label{ex:D5}
Consider a D5 brane on $\R^2 \times \mathbb{C}^2$ in $\IIB[\R^4_A \times \mathbb{C}^3_B]$. We then compute the Ext algebra
\begin{align*}
&\ul \ext _{\mr{Fuk}(\R^4 )\otimes \coh( \C^3 ) }(\R^2 \otimes \mc O_{\C^2}, \R^2 \otimes \mc O_{\C^2} ) \stackrel{\eqref{eqn:tensor}}{\cong}  \ul \ext _{\mr{Fuk}(\R^4 )  }(\R^2  , \R^2   )\otimes   \ul \ext _{  \coh( \C^3 ) }( \mc O_{\C^2}, \mc O_{\C^2}   )  \\
& \stackrel{\eqref{eqn:Fuk}, \eqref{eqn:coh}}{\cong}  \bracket{ \Omega^\bullet(\R^2) , d_{\R^2}} \otimes \bracket{ \Omega^{0,\bullet}(\C^2 ,  \wedge ^\bullet \mc N_{\C^2/\C^3 } ), \ol{\del} } \cong \bracket{ \Omega^\bullet(\R^2) \otimes  \Omega^{0,\bullet}(\C^2) [\eps],  d_{\R^2} + \ol{\del}_{\C^2}   }, 
\end{align*}
where $\eps$ is the degree 1 generator of the normal bundle, and $d_{\R^2} + \ol{\del}_{\C^2}  $ denotes $ d_{\R^2} \otimes 1+ 1\otimes  \ol{\del}_{\C^2}  $. This describes a Chern--Simons theory on $\R^2_{\mr{dR}}\times \C^2_{\ol{\del}}$ enhanced by $\C[\eps]$.
\end{ex}

\begin{rmk}\label{rmk:shift}
We have focused on open-string field theories in a CY category of dimension 5, which necessarily leads to $\Z/2\Z$-graded classical field theories. More precisely, in such a category, the cyclic pairing of the associated $L_\infty$-algebra has cohomological degree $\dim_\R M + \dim_\C X + 2\,\dim_\C Y - 5$. By  Definition~\ref{defn:genCS}, this forces the resulting theory to be only $\Z/2\Z$-graded.

However, one can produce a $\mathbb{Z}$-graded open-string field theory by choosing appropriate twisting data in the underlying supersymmetric field theory. Indeed, as discussed in \cite[Subsection~3.1.3]{RY1}, such a grading can arise from a shifted supergravity background normal to the branes. For example, for a D3-brane on $\R^2\times \C$ in $\IIB[\R^4_A \times \C^3_B]$, one may consider the background $\R^4 \times \C^2 \times \C[2]$, where $\C[2]$ is an affine space with a cohomological shift.\footnote{The same shift appears in the interpretation of the singular support condition in geometric Langlands, as discussed in \cite{EY2}, in which the moduli space of vacua of the field theory parametrizes possible configurations of D-branes in $\C[2]$.} With this shift, the geometry $\R^4 \times \C^2 \times \C[2]$ corresponds to a CY category of dimension 3, which aligns with many existing treatments of topological strings in the literature. Developing a supergravity framework that can fully account for a $\Z$-graded open-string field theory in general is beyond the scope of this paper.
\end{rmk}

\subsubsection{}

We will need to introduce compactification of open-string field theory along a compact direction, in order to obtain lower-dimensional theories.

\begin{defn}\cite[Section 19]{CostelloSUSY}
Let $\mc F$ be a classical field theory on $M$ and let $\pi \colon M\to N$ be a fiber bundle. The \emph{compactification}\footnote{We thank the referee for suggesting that compactification is a better term than dimensional reduction in this context.} of $\mc F$ along $\pi $ is defined as the pushforward $\pi _* \mc F$.
\end{defn}

When $ \pi \colon M=N\times F \to N$ is a trivial bundle with fiber $F$, then compactification along $\pi$ is also called compactification along $F$. In this case, the pushforward is computed by taking cohomology of $F$.

\begin{ex}\label{ex:CY2}
We illustrate the idea by examining a D5-brane on $\R^2 \times \C  \times \bb P^1 $ in $\IIB[\R^4_A \times (\C \times  \tot_{\bb P^1}(\mc O(-2)))_B ]$, resulting in a 6-dimensional field theory, and then by compactifying along $\bb P^1$.
 
  It yields an open-string field theory on $\R^2 \times \C  \times \bb P^1 $ defined by \begin{align*}
&\ul \ext _{\mr{Fuk}(\R^4 )\otimes \coh( \C \times  \tot_{\bb P^1}(\mc O(-2)) ) }(\R^2 \otimes \mc O_{\C  \times \bb P^1}, \R^2 \otimes \mc O_{\C  \times \bb P^1} ) \\
&\stackrel{\eqref{eqn:tensor}}{\cong}  \ul \ext _{\mr{Fuk}(\R^4 )  }(\R^2  , \R^2   )\otimes   \ul \ext _{  \coh(\C) }( \mc O_{\C   }, \mc O_{\C }   ) \otimes   \ul \ext _{  \coh(\tot_{\bb P^1}(\mc O(-2)) ) }( \mc O_{  \bb P^1}, \mc O_{\bb P^1}   )  \\
& \stackrel{\eqref{eqn:Fuk}, \eqref{eqn:coh}}{\cong}  \bracket{ \Omega^\bullet(\R^2) , d_{\R^2}} \otimes \bracket{ \Omega^{0,\bullet}(\C  )  , \ol{\del}_{\C  } } \otimes \bracket{ \Omega^{0,\bullet}(\bb P^1 , \wedge^\bullet \mc N_{\bb P^1 / \tot_{\bb P^1}(\mc O(-2))  } ), \ol{\del}  }.
\end{align*}
Let $\pi \colon \R^2 \times \C  \times \bb P^1  \to \R^2 \times \C $ be the projection.  To implement compactification along $\bb P^1$, we note \[H^\bullet ( \bb P^1, \wedge^\bullet \mc N_{\bb P^1 / \tot_{\bb P^1}(\mc O(-2))  }  )\cong H^\bullet(\bb P^1, \mc O) \oplus H^\bullet (\bb P^1, \mc O(-2))[-1]\cong \C \oplus \C \delta  \] where $\delta$ has cohomological degree 2 from the wedge power and cohomological grading. As a cyclic graded-commutative algebra, this is $\C[\delta]/(\delta^2)$ (with no higher products for a degree reason), with a trace map $\C[\delta]/(\delta^2)\to \C$ defined by $\delta\mapsto 1$. Hence the compactification of open-string field theory on a D5-brane on $\R^2 \times \C  \times \bb P^1 $ in $\IIB[\R^4_A \times (\C \times  \tot_{\bb P^1}(\mc O(-2)))_B ]$ along $\bb P^1$ is described by \[\left(   \Omega^\bullet(\R^2) \otimes   \Omega^{0,\bullet}(\C  )    \otimes \C[\delta]/(\delta^2) , d _{\R^2} + d_{\ol{\C}} \right) .\] This is a Chern--Simons theory on $\R^2_{\mr{dR}}\times \C_{\ol{\del}}$ enhanced by $\C[\delta]/(\delta^2) $.
\end{ex}

\subsection{Main Result}

\begin{prop}\label{prop:CSfromTS}
Each generalized Chern--Simons theory in Proposition \ref{prop:twistedSYMisCS}, which arises from a twist of pure supersymmetric Yang--Mills theory, can be realized as an open-string field theory in a suitable topological string background, as summarized in the following tables.
\end{prop}

\begin{rmk}
We note a parallel between our main theorem and Witten's seminal result on the realization of Chern--Simons theory via topological strings.	
\end{rmk}

In the table, we use $\varepsilon'$ to denote an odd variable obtained from compactification, as opposed to $\varepsilon$ variables arising from the embedding $\C^k \subset \C^{k+p}$ that parametrize a normal direction. Note that these variables are algebraically equivalent.

In the table, a closed string field $w$ renders the given open-string field theory perturbatively trivial, allowing the desired equivalences to follow immediately in those cases. Therefore, it suffices to address the remaining cases. The context for closed string fields will be explained further in Remark \ref{rmk:linear superpotential}.

\begin{table}[htbp]
\centering
\begin{tabular}{c|c||c|c} 
\text{Spacetime} & \({A}\) & \text{D-branes} & \text{Closed String} \\ \hline
\(\C^5_{\ol{\del}}\) & \(\C\) & 
{D9-branes on \(\C^5\)  in IIB[\(\C^5_B\)]} & 
\\ \hline

\(\R_{\mr{dR}} \times \C^4_{\ol{\del}}\) & \(\C\) & 
\shortstack{D8-branes on \(\R\times \C^4\) in IIA[\(\R^2_A\times \C^4_B\)]} & 
\\ \hline

\(\C^4_{\ol{\del}}\) & \(\C[\eps]\) & 
\shortstack{D7-branes on \(\C^4\)  in IIB[\(\C^5_B\)]} & 
\\ \hline
\(\R^2_{\mr{dR}} \times \C^3_{\ol{\del}}\) & \(\C\) & 
\shortstack{D7-branes on \(\R^2\times \C^3\)  in IIB[\(\R^4_A\times \C^3_B\)]} & 
\\ \hline
\(\C^4_{\ol{\del}}\) & \((\C[\eps],\,\frac{\del}{\del \eps})\) & 
{D7-branes on \(\C^4\)  in IIB[\(\C^5_B\)]} & 
{  \(w\in \PV(\C^5)\)} \\ \hline

\(\R_{\mr{dR}} \times \C^3_{\ol{\del}}\) & \(\C[\eps]\) & 
\shortstack{D6-branes on \(\R\times \C^3\) in IIA[\(\R^2_A\times \C^4_B\)]} & 
\\ \hline
\(\R^3_{\mr{dR}} \times \C^2_{\ol{\del}}\) & \(\C\) & 
\shortstack{D6-branes on \(\R^3\times \C^2\)  in IIA[\(\R^6_A\times \C^2_B\)]} & 
\\ \hline
\(\R_{\mr{dR}} \times \C^3_{\ol{\del}}\) & \( (\C[\eps],\,\frac{\del}{\del \eps} )\) & 
\shortstack{D6-branes on \(\R\times \C^3\)  in IIA[\(\R^2_A\times \C^4_B\)]} & 
\shortstack{ \(w\in \PV(\C^4)\)} \\ \hline

\(\C^3_{\ol{\del}}\) & \(\C[\eps_1,\eps_2]\) & 
\shortstack{D5-branes on \(\C^3\)  in IIB[\(\C^5_B\)]} & 
\\ \hline
\(\R^2_{\mr{dR}} \times \C^2_{\ol{\del}}\) & \(\C[\eps]\) & 
\shortstack{D5-branes on \(\R^2\times \C^2\)  in IIB[\(\R^4_A\times \C^3_B\)]} & 
\\ \hline
\(\R^4_{\mr{dR}} \times \C_{\ol{\del}}\) & \(\C\) & 
\shortstack{D5-branes on \(\R^4\times \C\)  in IIB[\(\R^8_A\times \C_B\)]} & 
\\ \hline
\(\R^2_{\mr{dR}} \times \C^2_{\ol{\del}}\) & \( (\C[\eps],\,\frac{\del}{\del \eps} )\) & 
\shortstack{D5-branes on \(\R^2\times \C^2\)  in IIB[\(\R^4_A\times \C^3_B\)]} & 
\shortstack{   \(w\in \PV(\C^3)\)} \\ \hline

\(\R_{\mr{dR}} \times \C^2_{\ol{\del}}\) & \(\C[\eps_1,\eps_2]\) & 
\shortstack{D4-branes on \(\R\times \C^2\)  in IIA[\(\R^2_A\times \C^4_B\)]} & 
\\ \hline
\(\R^3_{\mr{dR}} \times \C_{\ol{\del}}\) & \(\C[\eps]\) & 
\shortstack{D4-branes on \(\R^3\times \C\)  in IIA[\(\R^6_A\times \C^2_B\)]} & 
\\ \hline
\(\R^5_{\mr{dR}}\) & \(\C\) & 
\shortstack{D4-branes on \(\R^5\)  in IIA[\(\R^{10}_A\)]} & 
\\ \hline
\(\R^3_{\mr{dR}} \times \C_{\ol{\del}}\) & \( (\C[\eps],\,\frac{\del}{\del \eps} )\) & 
\shortstack{D4-branes on \(\R^3\times \C\)  in IIA[\(\R^6_A\times \C^2_B\)]} & 
\shortstack{  \(w\in \PV(\C^2)\)} \\ \hline

\(\C^2_{\mr{Dol}}\) & \(\C[\eps]\) & 
\shortstack{D3-branes on \(\C^2\) in IIB[\(\C^5_B\)]} & 
\\ \hline
\(\R^2_{\mr{dR}} \times \C_{\mr{Dol}}\) & \(\C[\eps]\) & 
\shortstack{D3-branes on \(\R^2\times \C\)  in IIB[\(\R^4_A\times \C^3_B\)]} & 
\\ \hline
\(\R^4_{\mr{dR}}\) & \(\C[\eps]\) & 
\shortstack{D3-branes on \(\R^4\)  in IIB[\(\R^8_A\times \C_B\)]} & 
\\ \hline
\(\C^2_{\mr{Dol}}\) & \( (\C[\eps],\,\frac{\del}{\del \eps} )\) & 
\shortstack{D3-branes on \(\C^2\)  in IIB[\(\C^5_B\)]} & 
\shortstack{ \(w\in \PV(\C^5)\)} \\ \hline
\(\R^2_{\mr{dR}} \times \C_{\mr{Dol}}\) & \( (\C[\eps],\,\frac{\del}{\del \eps} )\) & 
\shortstack{D3-branes on \(\R^2\times \C\)  in IIB[\(\R^4_A\times \C^3_B\)]} & 
\shortstack{  \(w\in \PV(\C^3)\)} \\ \hline
\(\R^4_{\mr{dR}}\) & \( (\C[\eps],\,\frac{\del}{\del \eps} )\) & 
\shortstack{D3-branes on \(\R^4\)  in IIB[\(\R^8_A\times \C_B\)]} & 
\shortstack{  \(w\in \PV(\C)\)} \\ \hline

\(\R_{\mr{dR}} \times \C_{\mr{Dol}}\) & \(\C[\eps_1,\eps_2]\) & 
\shortstack{D2-branes on \(\R\times \C\)  in IIA[\(\R^2_A\times \C^4_B\)]} & 
\\ \hline
\(\R^3_{\mr{dR}}\) & \(\C[\eps_1,\eps_2]\) & 
\shortstack{D2-branes on \(\R^3\)  in IIA[\(\R^6_A\times \C^2_B\)]} & 
\\ \hline
\(\R^3_{\mr{dR}}\) & \( (\C[\eps_1,\eps_2],\,\frac{\del}{\del \eps_1} )\) & 
\shortstack{D2-branes on \(\R^3\) in IIA[\(\R^6_A\times \C^2_B\)]} & 
\shortstack{   \(w_1\in \PV(\C^2)\)} \\ \hline

\end{tabular}
\caption{Chern--Simons theories enhanced by $A$ as open-string field theories using flat geometry}
\label{table1-new}
\end{table}

\begin{table}
\centering
\renewcommand{\arraystretch}{0.6}
\begin{tabular}{c|c||c|c}
\text{Spacetime} & \({A}\) & \text{D-branes and Compactification} & \text{Closed String} \\ \hline
\multirow{2}{*}{$\C^3_{\ol{\del}}$} & \multirow{2}{*}{$\C[\delta]/(\delta^2)$} & 
D7-branes on \(\C^3 \times \bb P^1\) in IIB\([ (\C^3 \times \tot_{\bb P^1}(\mc O(-2)))_B]\) & \\ 
 &  & compactified along \(\bb P^1\) & \\ \hline
\multirow{2}{*}{$\R_{\mr{dR}} \times \C^2_{\ol{\del}}$} & \multirow{2}{*}{$\C[\delta]/(\delta^2)$} & 
D6-branes on \(\R \times \C^2 \times \bb P^1\) in IIA\([\R^2_A \times (\C^2 \times \tot_{\bb P^1}(\mc O(-2)))_B]\) & \\
 &  & compactified along \(\bb P^1\) & \\ \hline
\multirow{2}{*}{$\C^2_{\ol{\del}}$} & \multirow{2}{*}{$\C[\eps,\delta]/(\delta^2)$} & 
D5-branes on \(\C^2 \times \bb P^1\) in IIB\([(\C^3 \times \tot_{\bb P^1}(\mc O(-2)))_B]\) & \\ 
 &  & compactified along \(\bb P^1\) & \\ \hline
\multirow{2}{*}{$\R^2_{\mr{dR}}\times \C_{\ol{\del}}$} & \multirow{2}{*}{$\C[\delta]/(\delta^2)$} & 
D5-branes on \(\R^2\times \C \times \bb P^1\) in IIB\([\R^4_A\times (\C \times \tot_{\bb P^1}(\mc O(-2)))_B]\) & \\ 
 &  & compactified along \(\bb P^1\) & \\ \hline
\multirow{2}{*}{$\C^2_{\ol{\del}}$} & \multirow{2}{*}{$\bigl(\C[\eps,\delta]/(\delta^2),\,\frac{\del}{\del \eps}\bigr)$} & 
D5-branes on \(\C^2 \times \bb P^1\) in IIB\([(\C^3 \times \tot_{\bb P^1}(\mc O(-2)))_B]\) & 
\multirow{2}{*}{\(w\in \PV(\C^3)\)} \\
 &  & compactified along \(\bb P^1\) & \\ \hline
\multirow{2}{*}{$\R_{\mr{dR}}\times \C_{\ol{\del}}$} & \multirow{2}{*}{$\C[\eps,\delta]/(\delta^2)$} & 
D4-branes on \(\R\times \C\times \bb P^1\) in IIA\([\R^2_A \times (\C^2 \times \tot_{\bb P^1}(\mc O(-2)))_B]\) & \\ 
 &  & compactified along \(\bb P^1\) & \\ \hline
\multirow{2}{*}{$\R^3_{\mr{dR}}$} & \multirow{2}{*}{$\C[\delta]/(\delta^2)$} & 
D4-branes on \(\R^3 \times \bb P^1\) in IIA\([\R^6_A \times \tot_{\bb P^1}(\mc O(-2))_B]\) & \\ 
 &  & compactified along \(\bb P^1\) & \\ \hline
\multirow{2}{*}{$\R_{\mr{dR}}\times \C_{\ol{\del}}$} & \multirow{2}{*}{$\bigl(\C[\eps,\delta]/(\delta^2),\,\frac{\del}{\del \eps}\bigr)$} & 
D4-branes on \(\R\times \C\times \bb P^1\) in IIA\([\R^2_A \times (\C^2 \times \tot_{\bb P^1}(\mc O(-2)))_B]\) & 
\multirow{2}{*}{\(w\in \PV(\C^2)\)} \\ 
 &  & compactified along \(\bb P^1\) & \\ \hline
\multirow{2}{*}{$\C_{\ol{\del}}$} 
  & \multirow{2}{*}{$\C[\eps_1,\eps_2,\delta]/(\delta^2)$} 
  & D3-branes on \(\C \times \bb P^1\) in IIB\([(\C^3 \times \tot_{\bb P^1}(\mc O(-2)))_B]\) & \\ 
  &  & compactified along \(\bb P^1\) & \\ \hline
\multirow{2}{*}{$\R^2_{\mr{dR}}$} 
  & \multirow{2}{*}{$\C[\eps,\delta]/(\delta^2)$} 
  & D3-branes on \(\R^2 \times \bb P^1\) in IIB\([\R^4_A \times (\C \times \tot_{\bb P^1}(\mc O(-2)))_B]\) & \\ 
  &  & compactified along \(\bb P^1\) & \\ \hline
\multirow{2}{*}{$\R^2_{\mr{dR}}$} 
  & \multirow{2}{*}{$\bigl(\C[\eps,\delta]/(\delta^2),\,\frac{\del}{\del \eps}\bigr)$} 
  & D3-branes on \(\R^2 \times \bb P^1\) in IIB\([\R^4_A \times (\C \times \tot_{\bb P^1}(\mc O(-2)))_B]\) 
  & \multirow{2}{*}{\(w\in \PV(\C)\)} \\ 
  &  & compactified along \(\bb P^1\) & \\ \hline
\end{tabular}
\caption{Chern--Simons theories enhanced by $A$ as open-string field theories using $\tot_{\bb P^1}(\mc O(-2))$}
\label{table2-new}
\end{table}

\begin{table}
\centering
\renewcommand{\arraystretch}{0.6}
\begin{tabular}{c|c||c|c}
\text{Spacetime} & \( {A}\) & \text{D-branes and Compactification} & Closed String \\ \hline
\multirow{4}{*}{$\C^2_{\ol{\del}}$} 
  & \multirow{4}{*}{$\C[\eps']$} 
  & D5-branes on \(\C^2 \times \bb P^1\) in IIB\([(\C^2 \times \tot_{\bb P^1}(\mc O(-1)\oplus \mc O(-1)))_B]\) \\ 
  && compactified along \(\bb P^1\) \\
\cline{3-4}
  &  & D7-branes on \(\C^2 \times \bb P^2\) in IIB\([(\C^2 \times \tot_{\bb P^2}(\mc O(-3)))_B]\) \\
  &&compactified along \(\bb P^2\) \\
   \hline
\multirow{4}{*}{$\R_{\mr{dR}}\times \C_{\ol{\del}}$} 
  & \multirow{4}{*}{$\C[\eps']$} 
  & D4-branes on \(\R \times \C \times \bb P^1\) in IIA\([\R^2_A \times (\C \times \tot_{\bb P^1}(\mc O(-1)\oplus \mc O(-1)))_B]\) \\ 
  &&compactified along \(\bb P^1\) \\
\cline{3-4}
  &  & D6-branes on \(\R \times \C \times \bb P^2\) in IIA\([\R^2_A \times (\C \times \tot_{\bb P^2}(\mc O(-3)))_B]\) \\
  &&compactified along \(\bb P^2\)\\
   \hline 
\multirow{4}{*}{$\C_{\ol{\del}}$}   & \multirow{4}{*}{$\C[\eps,\eps']$}  & D3-branes on \(\C \times \bb P^1\) in IIB\([(\C^2 \times \tot_{\bb P^1}(\mc O(-1)\oplus\mc O(-1)))_B]\) &\\
    &  & compactified along \(\bb P^1\) & \\ 
    \cline{3-4}
&&   D5-branes on \(\C \times \bb P^2\) in IIB\([(\C^2 \times \tot_{\bb P^2}(\mc O(-3)))_B]\)
  & \\ 
  &  & compactified along \(\bb P^2\) & \\ \hline
\multirow{4}{*}{$\R^2_{\mr{dR}}$} & \multirow{4}{*}{$\C[\eps']$} 
  &  D3-branes on \(\R^2 \times \bb P^1\) in IIB\([\R^4_A \times \tot_{\bb P^1}(\mc O(-1)\oplus\mc O(-1))_B]\) \\
  &  & compactified along \(\bb P^1\) & \\
    \cline{3-4}  
  && D5-branes on \(\R^2 \times \bb P^2\) in IIB\([\R^4_A \times \tot_{\bb P^2}(\mc O(-3))_B]\)  & \\ 
  &  & compactified along \(\bb P^2\) & \\ \hline
\multirow{4}{*}{$\C_{\ol{\del}}$} & \multirow{4}{*}{$\bigl(\C[\eps,\eps'],\,\frac{\del}{\del \eps}\bigr)$} 
  &  D3-branes on \(\C \times \bb P^1\) in IIB\([(\C^2 \times \tot_{\bb P^1}(\mc O(-1)\oplus\mc O(-1)))_B]\)& \multirow{4}{*}{\(w\in \PV(\C^2)\)} \\ 
  &  & compactified along \(\bb P^1\) & \\ 
    \cline{3-3}    
  && D5-branes on \(\C \times \bb P^2\) in IIB\([(\C^2 \times \tot_{\bb P^2}(\mc O(-3)))_B]\)
  &  \\ 
  &  & compactified along \(\bb P^2\) & \\ \hline
\end{tabular}
\caption{Chern--Simons theories enhanced by $A$ as open-string field theories using local CY 3-folds}
\label{table3-new}
\end{table}

\begin{table}
\centering
\renewcommand{\arraystretch}{0.6}
\begin{tabular}{c|c||c}
\text{Spacetime} & \( {A}\) & \text{D-branes and Compactification}   \\
\hline
  \multirow{2}{*}{\(\C_{\ol{\del}}\) }&  \multirow{2}{*}{\(\C[\delta_1,\delta_2]/(\delta_1^2,\delta_2^2)\)}  & D5-branes on $\C\times \bb P^1 \times \bb P^1$ in $\IIB[(\C \times \tot_{\bb P^1 \times \bb P^1}( \Omega^1_{\bb P^1 \times \bb P^1} ))_B ]$ \\  
&& compactified along $\bb P^1 \times \bb P^1 $ \\
\hline 
  \multirow{4}{*}{\(\C_{\ol{\del}}\) }&  \multirow{4}{*}{\(\C[\delta]/(\delta^2)\)}  & D5-branes on $\C\times \bb P^2$ in IIB[$(\C \times \tot_{\bb P^2}(\mc O(-1) \oplus \mc O(-2)   ) )_B$] \\  
&&  compactified along $\bb P^2$ \\
\cline{3-3}
 && D7-branes on $\C\times \bb P^3$ in  IIB[$(\C \times \tot_{\bb P^3}(\mc O(-4)   ) )_B$] \\  
&&  compactified along $\bb P^3$  \\
 \hline
\end{tabular}
\caption{Chern--Simons theories enhanced by $A$ as open-string field theories using local CY 4-folds}
\label{table4-new}
\end{table}
\FloatBarrier

We present the computations for $\ext_{\coh(X)}(\mc O_Z, \mc O_Z)$  required in the proof of Proposition \ref{prop:CSfromTS}.

\begin{lemma}\label{lemma:CYext}
The following are isomorphisms of cyclic graded-commutative algebras:
\begin{itemize}
	\item [(a)] $\ext_{\tot_{\bb P^1}(\mc O(-2)) }( \mc O_{\bb P^1} , \mc O_{\bb P^1} )\cong \C[\delta]/(\delta^2)$
	\item [(b)] $\ext_{\tot_{\bb P^1}(\mc O(-1) \oplus\mc O(-1) ) }( \mc O_{\bb P^1} , \mc O_{\bb P^1} )\cong \C[\eps']$
	\item [(c)] $\ext_{\tot_{\bb P^2}(\mc O(-3)) }( \mc O_{\bb P^2} , \mc O_{\bb P^2} )\cong  \C[\eps']$
	\item [(d)] $\ext_{\tot_{\bb P^1 \times \bb P^1}( \Omega^1_{\bb P^1 \times \bb P^1} )} ( \mc O_{\bb P^1 \times \bb P^1 },  \mc O_{\bb P^1 \times \bb P^1 } )\cong \C[\delta_1,\delta_2]/(\delta_1^2,\delta_2^2) $
	\item [(e)] $\ext_{\tot_{\bb P^2}(\mc O(-1) \oplus\mc O(-2) ) }( \mc O_{\bb P^2} , \mc O_{\bb P^2} )\cong \C[\delta]/(\delta^2)$
	\item [(f)] $\ext_{\tot_{\bb P^3}(\mc O(-4)   )}( \mc O_{\bb P^3},  \mc O_{\bb P^3} ) \cong \C [\delta]/(\delta^2) $
\end{itemize}
\end{lemma}

\begin{proof}
Since these computations are standard in algebraic geometry, we provide a brief outline.
\begin{itemize}
	\item [(a)] This was proved in Example \ref{ex:CY2}. Cases (c) and (f) follow by analogous arguments.
	\item [(b)]	The proof of (b) proceeds as follows. In general, for a smooth projective curve $C$ and a vector bundle $\mc V\to C$ of rank 2, we need to compute $H^\bullet(C , \bigoplus_{i=0}^2 \wedge^i \mc V ) $ as a cyclic graded-commutative algebra. The grading comes from both the wedge power and the cohomological degree: \[  \xymatrix@R-2pc{
\ul{0} & \ul{1} & \ul{2} & \ul{3}\\
H^0(C,\mc O_{C}) &H^1(C,\mc O_{C})  \oplus  H^0(C,\mc V)  & H^1(C,\mc V)  \oplus  H^0(C,\mc K_{C}) & H^1(C,\mc K_{C} )
}
 \] Moreover, the graded algebra structure $H^a( C, \wedge^b  \mc V ) \otimes H^c(C, \wedge^d  \mc V )\to H^{a+c}(C,  \wedge^{b+d} \mc  V )$ is induced by the wedge product, with a trace map $H^1(C , \mc K_{C}) \cong H^{1,1}(C) \cong \C$. In our case, we note 
\begin{align*}
 & \Ext_{  \tot _{\bb P^1}( \mc O(-1)\oplus\mc O(-1)  ) }( \mc O_{\bb P^1}, \mc O_{\bb P^1} ) \cong  H^\bullet \bracket{  \bb P^1, \bigoplus_{i=0}^2   \wedge^i ( \mc O(-1)\oplus\mc O(-1) ) }\\
 & \cong 	H^\bullet ( \bb P^1, \mc O ) \oplus H^\bullet (  \bb P^1,    \mc O(-1)\oplus\mc O(-1) )[-1] \oplus 	H^\bullet  ( \bb P^1, \mc O(-2) )[-2]  \cong  H^\bullet ( \bb P^1,\mc O  ) \oplus H^\bullet (\bb P^1,\mc O(-2))[-2]  \cong  \C \oplus \C \eps'  
\end{align*}
where $\eps'$ is a variable of cohomological degree 3. Thus, $\ext_{\tot_{\bb P^1}(\mc O(-1) \oplus\mc O(-1) ) }( \mc O_{\bb P^1} , \mc O_{\bb P^1} )$ is a graded-commutative algebra  $\C[\eps']$, equipped with an odd trace map $\tr \colon \C[\eps']\to \C$, defined by $\eps'\mapsto 1$. The case (e) follows from similar computation.
	\item [(d)] We observe that $\tot_{\bb P^1 \times \bb P^1}( \Omega^1_{\bb P^1 \times \bb P^1} ) \cong \tot_{\bb P^1}(\mc O(-2)) \times \tot_{\bb P^1}(\mc O(-2))$. Therefore, we obtain
	\[\ext_{\tot_{\bb P^1 \times \bb P^1}( \Omega^1_{\bb P^1 \times \bb P^1} )} ( \mc O_{\bb P^1 \times \bb P^1 },  \mc O_{\bb P^1 \times \bb P^1 } )\cong  \ext_{\tot_{\bb P^1}(\mc O(-2)) }( \mc O_{\bb P^1} , \mc O_{\bb P^1} )\otimes \ext_{\tot_{\bb P^1}(\mc O(-2)) }( \mc O_{\bb P^1} , \mc O_{\bb P^1} )\cong \C[\delta_1,\delta_2]/(\delta_1^2,\delta_2^2) \]
\end{itemize}
There are no higher multiplications for degree reasons.
\end{proof}

\begin{proof}[Proof of Proposition \ref{prop:CSfromTS}]
Let $X$ be a CY $(5-k-p-m)$-fold, which yields the topological string theory $\mr{II}[\R^{2m}_A \times ( \C^{k+p}\times X )_B]$. Let $Z$ be a compact, connected submanifold of $ X$. Consider a D-brane on $\R^m \times \C^k \times Z \subset \R^{2m} \times \C^{k+p} \times X$. The corresponding compactified open-string field theory is
\[\pi_* \left(  \ul \ext_{\mr{Fuk}(\R^{2m})\otimes  \coh(\C^{k+p}\times  X)}(\R^m\otimes  \mc O_{\C^k \times Z},\R^m \otimes  \mc O_{\C^k \times Z} )\right)  \cong \Omega^\bullet(\R^m)\otimes \Omega^{0,\bullet}(\C^k)\otimes \C[\eps_1,\cdots,\eps_p]\otimes \ext_{\coh(X)}(\mc O_Z,\mc O_Z )	\]
where $\pi \colon \R ^m\times \C ^k \times Z\to \R ^m\times \C ^k $ is the canonical projection. In other words, this construction determines the spacetime $\R^m_{\mr{dR}}\times \C^k_{\ol{\del}}$ and the associated graded-commutative algebra  \[A=  \C[\eps_1,\cdots,\eps_p]\otimes \ext_{\coh(X)}(\mc O_Z,\mc O_Z )\] 
with a cyclic pairing induced by the trace map $\C[\eps_1,\cdots,\eps_p]\to \C$, defined by $\eps_1\cdots\eps_p\mapsto 1$, and the CY structure on $X$. Since the only compact, connected submanifold of an affine space is a point, this construction is well-defined.

Table \ref{table1-new} follows from identifying the appearance of a Dolbeault stack. Specifically, the pair $(\text{spacetime}, A) = (\C_{\ol{\del}}, \C[\eps])$ can be identified with $(\C_{\mr{Dol}}, \C)$, since $\Omega^{0, \bullet}(\C)[\eps] \cong \Omega^{\bullet, \bullet}(\C)$ via the correspondence $\eps \mapsto dz$.\footnote{In more general settings, one must account for twisting homomorphisms or supergravity backgrounds.} Table \ref{table2-new} follows from Lemma \ref{lemma:CYext} (a), Table \ref{table3-new} from Lemma \ref{lemma:CYext} (b) and (c), and Table \ref{table4-new} from Lemma \ref{lemma:CYext} (d), (e), and (f).
\end{proof}

\begin{rmk}\label{rmk:linear superpotential}
By definition, a closed string field in a topological string theory determined by a CY category $\mc C$ is an element of the Hochschild cohomology\footnote{Technically, it is an element of the cyclic cohomology; however, we will disregard this subtlety.} $\HH^\bullet(\mc C)$ of $\mc C$. In our case, where $\mc C= \mr{Fuk}(\R^{2m}) \otimes \coh(\C^{k+p} \times X)$, we simplify the analysis by ignoring contributions from $\mr{Fuk}(\R^{2m})$ and $\coh(X)$, since no nontrivial elements from these factors are relevant here. A closed string field, viewed as an element of $\HH^\bullet(\coh(\C^{k+p}))$, induces a deformation of $A=\ext_{\coh( \C^{k+p})} ( \mc O_{\C^k}, \mc O_{\C^k } )$ via the \emph{closed-open map}\footnote{The Hochschild cohomology of a DG category $\mc C$ governs deformations of $\mc C$ and, correspondingly, deformations of any DG algebra $\ext_{\mc C}(\mc E, \mc E)$ for $\mc E \in \mc C$. This induces a closed-open map $\HH^\bullet(\mc C) \to \HH^\bullet(\ext_{\mc C}(\mc E, \mc E))$.}
\[ \Phi\colon  \HH^\bullet (\coh(\C^{k+p} ))\to \HH^\bullet(\ext_{\coh( \C^{k+p})} ( \mc O_{\C^k}, \mc O_{\C^k } )  )\cong  \HH^\bullet(\mc O(\C^{k|p})  ).\]
Concretely, writing $\mc O(\C^{k+p}) = \C[z_1, \dots, z_k, w_1, \dots, w_p]$ and applying the HKR theorem to identify Hochschild cohomology with polyvector fields, the closed-open map is described by a version of the Fourier transform
\[\xymatrix@R-15pt{\PV(\C^{k+p})  \ar[r]^-\Phi \ar@{=}[d] & \PV(\mc O(\C^{k|p}))\ar@{=}[d] \\
\C[z_1,\cdots,z_k,\del_{z_1},\cdots,\del_{z_k}, w_1,\cdots,w_p,\del_{w_1},\cdots,\del_{w_p}]  \ar[r] &  \C[z_1,\cdots,z_k,\del_{z_1},\cdots,\del_{z_k},\eps_1,\cdots,\eps_p, \del_{\eps_1},\cdots,\del_{\eps_p} ] }\] where the map is defined by $z_i \mapsto z_i$, $\del_{z_i} \mapsto \del_{z_i}$, $w_i \mapsto \del_{\eps_i}$, and $\del_{w_i} \mapsto \eps_i$. Thus, a closed string field $w_i \in \PV(\C^{k+p})$ induces a deformation of the open-string field theory on $\C^k \subset \C^{k+p}$, corresponding to a differential $\del_{\eps_i}$ on $\Omega^{0, \bullet}(\C^k)[\eps_1, \dots, \eps_p]$. Such a closed string field is known as a \emph{linear superpotential}. For further details, see the original paper of Costello and Li \cite[Subsection 7.2]{CL16}.
\end{rmk}

\subsection{Remarks}

We provide several remarks to give additional context for the result.

\subsubsection{}

For each twist of pure supersymmetric Yang--Mills theory, we have presented one or two specific D-brane configurations, using only projective spaces and the total spaces of vector bundles over them. This choice was motivated by the importance of toric CY manifolds in string dualities. That is, while some of our configurations may not have known untwisted counterparts in physical string theory, we expect them to be connected to geometric engineering setups through a sequence of dualities.

On the other hand, many other D-brane configurations can also realize the same cyclic graded-commutative algebra $A$, since $A$ depends solely on cohomological data.

\begin{ex}
Consider
\[ \ul\ext_{\C^{k+p}} ( \mc O_{\C^k },  \mc O_{\C^k } ) \cong \Omega^{0,\bullet}(\C^k)[\eps_1,\cdots,\eps_p] \]
On the other hand, for an elliptic curve $E$, we have 
\[ \pi_*  \bracket{  \ul\ext_{\C^{k+p-1} \times E } ( \mc O_{\C^k\times E  } ,  \mc O_{\C^k\times E  } )} \cong \Omega^{0,\bullet}(\C^k)[\eps_1,\cdots,\eps_{p-1}] \otimes H^\bullet(E,\mc O_E) \]
where $\pi\colon \C^k \times E\to \C^k$ is the projection map. This implies that the open-string field theory of a D-brane on $\C^k \subset \C^{k+p}$ is equivalent to the compactification of the open-string field theory of a D-brane on $\C^k \times E \subset \C^{k+p-1} \times E$ along $E$.
\end{ex}

\begin{rmk}
This is consistent with physical expectations: compactification along an elliptic curve preserves supersymmetry, as reflected by the fact that both $H^\bullet(E, \mc O_E)$ and $\ext_{\C^2}(\mc O_\C, \mc O_\C)$ coincide with $\C[\eps]$.
\end{rmk}

In the following example, two different configurations (one involving a complex string background with a vector bundle defined by a nontrivial extension of line bundles) still give rise to the same cyclic graded-commutative algebra:
 \begin{ex}
Let $E$ be an elliptic curve, and let $\mc L$ be a nontrivial line bundle of degree 0. Consider the vector bundle $\mc V = \mc L \oplus \mc L^{-1}$ over $E$. D5-branes on $\C^2 \times E \subset \C^2 \times \tot_E \mc V$ in the theory $\IIB[(\C^2 \times \tot_E \mc V)_B]$ produce the holomorphic twist of a 4-dimensional $\mc N = 2$ pure gauge theory upon compactification along $E$. On the other hand, if $\mc L$ is isomorphic to its dual, then there exists a nontrivial vector bundle $\mc V$ of rank 2, classified by a nonzero element in $\ext^1(\mc L, \mc L) \cong H^1(E, \mc O_E)$. In this case, D5-branes on $\C^2 \times E \subset \C^2 \times \tot_E \mc V$ still realize the same theory after compactification along $E$. 
\end{ex}

\subsubsection{}

Supersymmetric Yang--Mills theories in dimensions $\geq 3$ arise as compactifications of 10-dimensional, 6-dimensional, and 4-dimensional supersymmetric Yang--Mills theories. Accordingly, they can all be understood as originating from string backgrounds such as IIB[$\C^5_B$], IIB[$(\C^3 \times \tot_{\bb P^1}(\mc O(-2)))_B$], and IIB[$(\C^2 \times \tot_{\bb P^2}(\mc O(-3)))_B$]. As an example, we illustrate how IIB$[\C^5_B]$ can give rise to the other backgrounds listed in Table \ref{table1-new}.
\begin{itemize}
	\item Consider $\C^2_{z_1, z_2}$ with the holomorphic volume form $\Omega = dz_1 \wedge dz_2$ and its inverse bivector $\Pi = \Omega^{-1} = \del_{z_1} \wedge \del_{z_2}$. There is a quasi-isomorphism between the Poisson complex $(\PV(\C^2), [\Pi, -]_{\mr{SN}})$ and the de Rham complex $(\Omega^\bullet(\R^4), d)$. This indicates that IIB$[\R^4_A \times \C^3_B]$ can be viewed as a deformation of IIB$[\C^5_B]$ by a closed string field $\Pi = \del_{z_1} \wedge \del_{z_2}$. Under this deformation, the open-string field theory of a D-brane on $\C^{k-1}_{z_1, z_3, \dots, z_k}$ is transformed into one on $\R^2 \times \C^{k-2}_{z_3, \dots, z_k}$, turning a purely holomorphic twist into a holomorphic-topological twist. For example, a D3-brane world-volume theory on $\C_{z_1} \times \C_{w_1}$ in IIB[$\C^5_B$] becomes a theory on $\R^2 \times \C_{w_1}$ in IIB$[\R^4_A \times \C^3_B]$. Further deformations of IIB$[\R^4_A \times \C^3_B]$ lead to IIB$[\R^8_A \times \C_B]$. For more details, see \cite[Subsection 6.6]{CL16}.

	\item There is an equivalence between the wrapped Fukaya category $\mr{Fuk}_{\mc W}(T^*S^1)$ of $T^*S^1$ and the category $\coh(\C^\times)$ of coherent sheaves on $\C^\times$ \cite{AAEKO}. This equivalence induces a T-duality between IIB$[(\C^\times \times \C^4)_B]$ and IIA$[(T^*S^1)_A \times \C^4_B]$. Under this duality, open-string field theories on $\C^\times \times \C^k$ and $\pt \times \C^k$ are mapped to those on $T^*_{\pt} S^1 \times \C^k$ and $S^1 \times \C^k$, respectively. For more on T-duality, see \cite[Subsection 4.2]{RY1}.
\end{itemize}
By combining these two points, we conclude that any string background in Table \ref{table1-new} can be derived from IIB$[\C^5_B]$. Consequently, the corresponding supersymmetric Yang--Mills theories can, in principle, all be realized starting from IIB$[\C^5_B]$, as expected.

\subsubsection{}

We have considered two types of deformations:
\begin{itemize}
	\item A linear superpotential, which deforms the open-string field theory via the differential $\frac{\del}{\del \eps}$, and
	\item A Poisson bivector field $\Pi$, which turns the holomorphic dependence of the spacetime of a open-string field theory into a topological one.
\end{itemize}
These deformations precisely correspond to residual supersymmetries of the type IIB supersymmetry algebra in 10 dimensions (see \cite[Subsection 6.6]{CL16}).

According to \cite[Subsections 5.1, 5.2]{RY1}, the geometric Langlands correspondence arises from twisted S-duality between these two deformations. Specifically, consider the open-string field theory on $\R^2 \times \C$ in IIB[$\R^4_A \times \C^3_B$], corresponding to the $d=4$, $\mc N=4$ holomorphic-topological twist. Twisted S-duality provides a correspondence between the deformation by $\frac{\del}{\del \eps}$ and the deformation by $\Pi$, leading to a conjectural equivalence between the categories of boundary conditions for these theories. This equivalence is precisely the geometric Langlands correspondence. Furthermore, \cite{EY3} argues that the quantum geometric Langlands correspondence also arises from this twisted S-duality. This perspective is conceptually simpler than the original proposal by \cite{KW} (see \cite[Remark 1.11]{EY3}).

In fact, there are infinitely many closed string fields that can induce further deformations of open-string field theory. For example, \cite[Subsections 5.3, 5.4]{RY1} discusses superconformal deformations and Omega backgrounds, and realizes $d=4$ super Chern--Simons theory as a further deformation of the $d=4$, $\mc N=4$ holomorphic-topological twist. Exploring these and other deformations within the framework of topological strings offers promising directions for future research.

\subsubsection{}

Topological strings have a rich literature --- both in their own right and in relation to physical string theory. In this article, we have presented a simple framework for engineering various field theories of interest, many of which, to our knowledge, are new to the literature.

Although this approach omits much of the rich structure of physical string theory, it opens new opportunities: viewing topological strings as twists of the full theory can reveal novel dualities and correspondences, as illustrated by the case of $d=4$, $\mc N=4$ holomorphic-topoligcal twist. 

This suggests that our approach is not merely a ``poor man's version'' of string theory but instead highlights structures particularly relevant to certain areas of mathematics. We hope to explore these and other ideas further in future work.


\end{document}